\documentclass[submission,copyright,creativecommons,sort&compress]{eptcs}
\providecommand{\event}{FROM 2026} 
\usepackage[numbers,sort&compress]{natbib}
\usepackage{iftex}

\ifpdf
  \usepackage{underscore}         
  \usepackage[T1]{fontenc}        
\else
  \usepackage{breakurl}           
\fi

\title{Dependently Typed Model Composition for Matching Logic}

\author{\'{A}d\'{a}m Kurucz
\institute{ELTE E\"{o}tv\"{o}s Lor\'{a}nd University\\
Budapest, Hungary}
\email{cphfw1@inf.elte.hu}
\and
P\'{e}ter Bereczky
\institute{ELTE E\"{o}tv\"{o}s Lor\'{a}nd University\\
Budapest, Hungary}
\email{berpeti@inf.elte.hu}
\and
D\'{a}niel Horp\'{a}csi
\institute{ELTE E\"{o}tv\"{o}s Lor\'{a}nd University\\
Budapest, Hungary}
\email{daniel-h@elte.hu}
}

\def\titlerunning{Dependently Typed Model Composition for Matching Logic}
\def\authorrunning{\'{A}. Kurucz, P. Bereczky \& D. Horp\'{a}csi}

\usepackage{mathtools,amssymb,amsmath}
\usepackage{cleveref}
\usepackage{array,multirow}
\usepackage{ebproof}
\usepackage{amsthm}
\usepackage[page]{appendix}

\usepackage{tikz}
\usetikzlibrary{arrows.meta}
\usetikzlibrary{calc}
\usetikzlibrary{fit}
\usetikzlibrary{shapes.geometric}
\usetikzlibrary{matrix}
\usetikzlibrary{tikzmark}

\NewDocumentCommand{\smallcase}{m m m}{
	\biggl\{
	\begin{array}{@{} l @{\mskip 10mu\relax} l @{}}
		#1 & \text{if } #2 \\
		#3 & \text{otherwise}
	\end{array}
}
\NewDocumentCommand{\tinynegspace}{}{
    \mskip -6mu\relax
}

\ExplSyntaxOn

\cs_new_eq:NN \old_texttt:n \texttt
\cs_set:Npn \texttt #1 {\textup{\old_texttt:n{#1}}}

\prop_const_from_keyval:Nn \preamble_dictionary {
    HList=\texttt{HList},
	partial=\texttt{Partial},
    iso=\texttt{Iso},
	partialf=\textit{to},
	partialg=\textit{from},
	partialfg=\textit{tofrom},
	partialgf=\textit{fromto},
	gluable=\texttt{Combinable},
    option=\texttt{Option},
	subone=\textit{sub}^{\textit{sorts}}\c_math_subscript_token1,
	subtwo=\textit{sub}^{\textit{sorts}}\c_math_subscript_token2,
	subfive=\textit{sub}^{\textit{symb}}\c_math_subscript_token1,
	subsix=\textit{sub}^{\textit{symb}}\c_math_subscript_token2,
	subtotalone=\textit{no_junk}^{\textit{sorts}},
	subtotalthree=\textit{no_junk}^{\textit{symb}},
    symbols=\textit{Symb},
	liftone=\textit{lift},
	lifttwo=\textit{lift},
	unliftone=\textit{unlift},
	unlifttwo=\textit{unlift},
    In=\texttt{In},
    ineq=\textit{in_eq},
    incons=\textit{in_cons},
    Inmap=\textit{In_map},
    Inmapr=\textit{In_map_r},
    SharedSort=\textit{UnitedSorts},
    SharedSymbols=\textit{UnitedSymb},
    agreeone=\textit{EV_agree},
    agreetwo=\textit{SV_agree},
    agreethree=\textit{args_agree},
    agreefour=\textit{ret_agree},
    agreefive=\textit{carrier_agree},
    agreesix=\textit{app_agree},
    glue=\textit{combine},
    glueT=\textit{combineT},
    gluableM=\texttt{Gluable},
    toinj=\textit{to_inj},
    metanat=\ensuremath{\mathbb N},
    metabool=\ensuremath{\mathbb B},
    metapos=\ensuremath{\mathbb N^+},
    metatrue=\textit{true},
    metafalse=\textit{false},
    gnum=\textit{num},
    gcbool=\textit{c_bool},
    gcisEven=\textit{c_isEven},
    gbool=\textit{bool},
    gisEven=\textit{isEven},
    gtwice=\textit{twice},
    gthrice=\textit{thrice},
    gplus=\textit{plus},
    hlistmapmove=\textit{HList_map_move},
    defaultsum=\ensuremath{C^{\textit{def}}_\uplus},
    inlpartial=\ensuremath{P_{\textit{\dict{inl}}}},
    inrpartial=\ensuremath{P_{\textit{\dict{inr}}}},
    deval=\ensuremath{\textit{unlift}^{\textit{eval}}},
    extendable=\texttt{Extendable},
    supplone=\textit{EV_ext},
    suppltwo=\textit{SV_ext},
    supplthree=\textit{args_ext},
    supplfour=\textit{ret_ext},
    List=\texttt{List}
}

\seq_new:N \preamble_appendix

\NewDocumentEnvironment{addtoappendix}{+b}{
    \seq_gput_right:Nn \preamble_appendix {#1}
}{}

\NewDocumentCommand{\renderappendix}{}{
    \crefalias{section}{appendix}
    \renewcommand{\appendixpagename}{Appendix}
    \clearpage
    \begin{appendices}
    \seq_use:Nn \preamble_appendix {}
    \end{appendices}
}

\tl_const:Nn \ld {.\ }

\NewDocumentCommand{\la}{m}{
	\lambda\:\clist_use:nn {#1} {\ } \ld
}

\NewDocumentCommand{\fa}{m}{
	\forall\:\clist_use:nn {#1} {\ } \ld
}

\NewDocumentCommand{\te}{m}{
	\exists\:\clist_use:nn {#1} {\ } \ld
}

\NewDocumentCommand{\dict}{m}{
	\prop_get:NnNTF \preamble_dictionary {#1} \l_tmpa_tl {\l_tmpa_tl} {\textit{#1}}
}

\NewDocumentCommand{\app}{m m}{
	#1 \  \clist_use:nn {#2} {\ }
}

\NewDocumentCommand{\sephetlist}{m}{
	\langle \clist_use:nn {#1} {,\ } \rangle
}

\tl_const:Nn \breakrow {\\ &\qquad}

\NewDocumentCommand{\typedef}{O{inductive} m o m}{
	\noindent\begin{tabular}{>{\(}l<{\)} @{\;} >{\(}l<{\)} @{}}
		\multicolumn{2}{@{} >{\(}l<{\)} }{\textbf{#1}\ #2\IfValueT{#3}{: #3}}
		\clist_map_inline:nn {#4} {\\ \mskip 8mu\relax \use_i:nn ##1 &: \use_ii:nn ##1}
	\end{tabular}
}

\NewDocumentCommand{\listof}{m}{
    \app{\dict{List}}{#1}
}

\NewDocumentCommand{\setnobreak}{}{
    \bool_set_true:N \l_tmpb_bool
}

\NewDocumentCommand{\fundef}{s t! m m m}{
    \begin{gather*}
        \begin{aligned}
            #3 &: #4
        \end{aligned} \\
        \begin{\IfBooleanTF{#2}{gathered}{aligned}}
            \bool_set_true:N \l_tmpa_bool
            \clist_map_inline:nn {#5} {
                \bool_if:NF \l_tmpa_bool {
                    \IfBooleanTF{#2}{
                        \bool_if:NTF \l_tmpb_bool {\qquad\qquad} {\\}
                    }{\\}
                }
                \bool_set_false:N \l_tmpa_bool
                \bool_set_false:N \l_tmpb_bool
                \IfBooleanTF{#1}{
                    \app{\use_i:nnnn ##1}{(\app{#3}{\use_ii:nnnn ##1}), \use_iii:nnnn ##1} \IfBooleanF{#2}{&}\coloneqq \use_iv:nnnn ##1
                }{
                    \app{#3}{\use_i:nn ##1} \IfBooleanF{#2}{&}\coloneqq \use_ii:nn ##1
                }
            }
        \end{\IfBooleanTF{#2}{gathered}{aligned}}
    \end{gather*}
}

\NewDocumentCommand{\map}{m m}{\app{\dict{map}}{#1, #2}}

\NewDocumentCommand{\case}{m o m s o m}{
    \app{\dict{case}}{#1, \IfValueTF{#2}{(\la {#2} #3)}{#3}\IfBooleanT{#4}{\breakrow}, \IfValueTF{#5}{(\la {#5} #6)}{#6}}
}

\NewDocumentCommand{\sigsortedclosedpattern}{m m m m}{\app{\dict{Pattern}}{#1, #3, #4, #2}}

\NewDocumentCommand{\eqby}{m}{=(\textit{by~}#1)}

\NewDocumentCommand{\DPair}{m m m}{
    \{#2 : #1 \:\&\: #3\}
}

\NewDocumentEnvironment{proofsketch}{o +b}{
    \noindent\textit{Proof~sketch.}\ \ \ \ #2
}{}

\NewDocumentCommand{\topset}{m}{\top\c_math_subscript_token{#1}}

\NewDocumentCommand{\othercase}{m m m o}{
\begin{cases}
    #1 & \text{if~} #2 \\
    #3 & \IfValueTF{#4}{\text{if~} #4}{\text{otherwise}}
\end{cases}
}

\tl_const:Nn \unprovedlemma {proposition}

\NewDocumentCommand{\enumcons}{m}{
    \clist_use:nn {#1} {\mid}
}

\NewDocumentCommand{\mapthese}{m}{
    \bool_set_true:N \l_tmpa_bool
    \{ \clist_map_inline:nn {#1} {\bool_if:NF \l_tmpa_bool {,\ } \bool_set_false:N \l_tmpa_bool \use_i:nn ##1 \mapsto \use_ii:nn ##1} \}
}

\ExplSyntaxOff

\newcommand{\lfp}{\textbf{lfp}}

\newcommand{\bool}{\textit{bool}}
\newcommand{\Coq}{Rocq}

\newcommand{\argsorts}{params}
\newcommand{\retsorts}{return}
\newcommand{\placeholder}{\square}

\NewDocumentCommand{\sepend}{m m m s}{#3\IfBooleanTF{#4}{#2}{#1\sepend{#1}{#2}}}
\newcommand{\separator}[1]{\sepend{#1}{}}
\newcommand{\polyapp}[1]{#1(\sepend{,\ }{)}}
\newcommand{\sepcomma}{\separator{,\ }}
\newcommand{\seplist}{[\sepend{,\ }{]}}

\newcommand{\ofsortvar}[2]{{#1}:{#2}}
\newcommand{\ofsort}[2]{#1_{#2}}
\newcommand{\update}[3]{#1[#2 \mapsto #3]}

\newcommand{\subst}[3]{#1[#3/#2]}
\NewDocumentCommand{\defined}{m o O{#2}}{\lceil #1 \rceil \IfValueT{#2}{_{#2}^{#3}}}
\NewDocumentCommand{\total}{m o O{#2}}{\lfloor #1 \rfloor \IfValueT{#2}{_{#2}^{#3}}}
\newcommand{\bevar}[1]{\underline{#1}}
\newcommand{\bsvar}[1]{\underline{\underline{#1}}}
\newcommand{\fevar}[1]{\widehat{#1}}
\newcommand{\fsvar}[1]{\widehat{\widehat{#1}}}

\newcommand{\syapp}[2]{#1 \cdot #2}

\newcommand{\closedpattern}[2]{\app {\texttt{Pattern}}{#1, #2}}
\newcommand{\sortedclosedpattern}[3]{\app {\texttt{Pattern}}{#2, #3, #1}}

\NewDocumentCommand{\EVar}{o}{\IfValueTF{#1}{\ofsort{EV}{#1}}{EV}}
\NewDocumentCommand{\SVar}{o}{\IfValueTF{#1}{\ofsort{SV}{#1}}{SV}}

\newcommand{\fst}[1]{#1.1}
\newcommand{\snd}[1]{#1.2}
\NewDocumentCommand{\FV}{o m}{FV\IfValueT{#1}{_{#1}}(#2)}
\newcommand{\conscc}{\mathbin{:\!:}}
\newcommand{\fmap}{\mathbin{<\mspace{-6mu}\$\mspace{-6mu}>}}
\newcommand{\hlistmap}[2]{#1 \fmap #2}

\newtheorem{lemma}{Lemma}

\newtheorem{proposition}{Proposition}
\newtheorem{definition}{Definition}
\newtheorem*{remark*}{Remark}

\begin{document}
\maketitle

\begin{abstract}
This paper investigates model composition---often referred to as ``gluing''---within the framework of matching logic. Specifically, we examine the systematic combination of existing signatures, variable valuations, theories, and their corresponding models. Our primary objective is to ensure that this composition preserves satisfaction proofs; that is, any theory validated by the individual constituent models is also validated by the resulting composite model. Our definitions are based on a polyadic, sorted variant of matching logic, which has also been expressed in the \Coq{} proof assistant. Therefore, we outline our definitions with dependent types for this work to serve as a direct blueprint for the mechanization in the short-term future.
\end{abstract}

\section{Introduction}\label{sec:intro}

Matching logic provides the theoretical foundation of the $\mathbb{K}$ semantic framework~\cite{Kfw}. Every programming language defined in $\mathbb{K}$ can be automatically translated into a corresponding matching logic theory. A fundamental property of any such theory is consistency: it must be guaranteed that the axioms defining the language theory do not introduce contradictions. In the model-theoretic approach, consistency can be established by showing that the theory admits a satisfying model, since any theory with a model is necessarily consistent.

In practice, matching logic theories specified in $\mathbb{K}$ are constructed by composing multiple so-called modules. Different sublanguages of a programming language are defined in separate modules, each corresponding to a matching logic theory, and these modules can be imported into one another. Common language components, such as standard containers (e.g., lists and maps) and standard syntactic categories (e.g., identifiers, booleans, and integers), are typically provided as reusable standard modules and incorporated into language definitions through $\mathbb{K}$'s import mechanism. The consistency of theories constructed this way can be verified compositionally: once the consistency of theories generated from standard modules has been established, the corresponding arguments can be reused when verifying larger theories that depend on these modules. From a model-theoretic perspective, this amounts to constructing a satisfying model for a composite theory from satisfying models of its constituent theories.

The primary motivation for this work is the verification of consistency for such $\mathbb{K}$ theories. Although the Kore intermediate language used in the $\mathbb{K}$ framework (and in our \Coq{} embedding~\cite{kore-ml}) is richer than matching $\mu$-logic, we restrict our presentation to the original matching $\mu$-logic formalism~\cite{chen2019mu}, as it captures all the essential challenges involved in compositional consistency verification.

In this paper, we define gluing and briefly discuss extension operations (inspired by the disjoint union) for compatible matching logic signatures, theories, and models, ensuring that shared components remain consistent. The gluing process combines these theories and models into a unified domain while preserving their satisfaction proofs, thereby establishing the consistency of the resulting theories. We express our results in a variant of type theory providing a direct blueprint for implementation in \Coq{}~\cite{coq} based on our existing formalization~\cite{kore-ml}, enabling high-assurance verification of theory consistency.

Contrary to related work~\cite{LTLMC} (based on an applicative, unsorted variant of matching logic~\cite{chen2021explained}), we use an intrinsically sorted variant of matching logic endowed with polyadic symbol applications, which significantly reframes the challenges of this process. Due to inherent sortedness, the logic does not need to be extended to accommodate sorted quantification. Furthermore, combined with polyadic applications, it also simplifies the interpretation of symbols by eliminating ``junk'' values resulting from partial or ill-sorted applications. Finally, since the implementation~\cite{kore-ml} this work is based on treats definedness and injections as built-in syntactical constructs, rather than theories included in nearly every model, they require no special treatment, and are therefore omitted from this presentation for clarity and brevity.

\section{Matching logic}\label{sec:background}

In this section, we recall matching logic and its implementation with dependent types. For a longer discussion on dependent types we refer to \Cref{sec:depintro}.

\begin{addtoappendix}
\section{A brief introduction to dependent types}\label{sec:depintro}

Let us introduce the metatheory used throughout the paper by presenting several common data type definitions in it. We formalize our framework in a variant of the Calculus of Inductive Constructions (CIC)---the dependent type theory that serves as the foundation for the \Coq{} proof assistant—featuring native inductive definitions. To streamline the presentation, the predicative universe hierarchy is left implicit, allowing $\dict{Type}$ to represent types at any level~\cite{pfenning-mohring}.\footnote{We do not rely on the inconsistent type-in-type axiom, as the universe levels are still inferred and enforced by \Coq{}~\cite{coqrefman} in the formalization.}

First, we present three well-known types: lists, union (also called sum), and option types.

\medskip\hfill
\begin{minipage}{0.35\textwidth}
\noindent
\typedef[inductive]{\app{\texttt{List}}{A}}[\dict{Type}]{
    {\textit{nil}}{\app {\texttt{List}}{A}},
    {\textit{cons}}{A \to \app{\texttt{List}}{A} \to \app{\texttt{List}}{A}}
}  
\end{minipage}\hfill
\begin{minipage}{0.3\textwidth}
\noindent
\typedef[inductive]{\app{\texttt{Sum}}{A,B}}[\dict{Type}]{
    {\textit{inl}}{A \to \app {\texttt{Sum}}{A, B}},
    {\textit{inr}}{B \to \app {\texttt{Sum}}{A, B}}
}    
\end{minipage}\hfill
\begin{minipage}{0.3\textwidth}
\noindent
\typedef[inductive]{\app{\dict{option}}{A}}[\dict{Type}]{
    {\dict{Some}}{A \to \app {\dict{option}}{A}},
    {\dict{None}}{\app {\dict{option}}{A}}
}    
\end{minipage}\hfill
\medskip

Lists are either empty or non-empty. We use $[]$ for \textit{nil} (i.e. the empty list) and the infix notation $\conscc$ for \textit{cons}, which prepends an element to a list. We use the standard notation $[e_1, \dots, e_n]$ to denote finite lists with their elements, and $\map{f}{\textit{xs}}$ to denote the application of function $f$ to all elements of the list \textit{xs}.

Elements of the union type (also denoted by $\uplus$) are tagged, \textit{inl} is used to inject elements of the first type into the union, while \textit{inr} is used to tag elements of the second type. The \texttt{Option} type extends a base type $A$ by wrapping its values in the \dict{Some} constructor and introducing an explicit error value, \dict{None}. We continue by presenting vectors (length-indexed lists) and dependent pairs.

\medskip\hfill
\begin{minipage}{0.5\textwidth}
\noindent
\typedef[inductive]{\app {\texttt{Vec}}{A}}[\mathbb N \to \dict{Type}]{
	{\textit{vnil}}{\app {\texttt{Vec}}{A, 0}},
	{\textit{vcons}}{\fa {n} A \to \app{\texttt{Vec}}{A, n} \to \app {\texttt{Vec}}{A, (1 + n)}}
}
\end{minipage}\hfill
\begin{minipage}{0.4\textwidth}
\noindent
\typedef[inductive]{\app {\texttt{SigT}}{A, (P : A \to \dict{Type})}}[\dict{Type}]{
	{\textit{existT}}{\fa{x} \app {P}{x} \to \app {\texttt{SigT}}{A, P}}
}
\end{minipage}\hfill
\medskip

Vectors closely resemble lists, only their type differs: the vector's length (a natural number) appears explicitly as a type index of the $\texttt{Vec}$ type. Consequently, the length index of an empty vector is zero, whereas the length index of a non-empty vector is the successor of its tail's length.

Dependent pairs are a generalization of pairs: they have a single constructor with two arguments, but the second argument can be defined in terms of the first one. We reuse the usual notation of pairs to denote dependent pairs, and use the notation $\DPair{A}{x}{\app{P}{x}}$ to denote the type $\app{\texttt{sigT}}{A, P}$. For example, $\app{\textit{existT}}{0,\textit{vnil}}$ is denoted by $(0, \textit{vnil})$ and $\app{\textit{existT}}{1,(\app{\textit{vcons}}{\textit{true}, \textit{vnil}})}$ by $(1, \app{\textit{vcons}}{\textit{true}, \textit{vnil}})$ both of which are of type $\DPair{\mathbb{N}}{x}{\app{\texttt{Vec}}{\mathbb{B}, x}}$, where the type index of the boolean-valued vector in the second component must be equal to the number specified in the first component. Furthermore, we use $\fst{x}$ and $\snd{x}$ to denote the first and second components for values $x$ of both pairs and dependent pairs.

Inductive types can also be used to express judgments, for example, list membership.

\medskip\hfill
\begin{minipage}{0.6\textwidth}
    \noindent
    \typedef{\app{\dict{In}}{A}}[A \to \listof{A} \to \dict{Type}]{
        {\dict{ineq}}{\fa {x, xs} \app{\dict{In}}{A, x, (x \conscc xs)}},
        {\dict{incons}}{\fa {x, y, xs} \app{\dict{In}}{A, y, xs} \to \app{\dict{In}}{A, y, (x \conscc xs)}}
    }
\end{minipage}\hfill
\medskip

The two constructors establish membership within a non-empty list by asserting that an element is either the head or resides within the tail. Crucially, elements of type \texttt{In} represent more than structural proofs of membership; they additionally serve as indices. Specifically, the number of nested \dict{incons} constructors corresponds directly to the included element's precise index. For this reason, we utilize them as de Bruijn indices in \Cref{sec:dependent}.

Next, we define the type of heterogeneous lists, which can contain elements of varying types. These lists are utilized in \Cref{sec:dependent}.

\medskip
\noindent\hfill
\begin{minipage}{0.6\textwidth}
\typedef[inductive]{\app{\texttt{HList}}{A, (F : A \to \dict{Type})}}[\listof{A} \to \dict{Type}]{
    {\textit{hnil}}{\app {\texttt{HList}}{A, F, []}},
    {\textit{hcons}}{\fa {x, \textit{xs}} \app{F}{x} \to \app {\texttt{HList}}{A, F, \textit{xs}} \to \app {\texttt{HList}}{A, F, (x \conscc \textit{xs})}}
}    
\end{minipage}\hfill
\medskip

Heterogeneous lists are indexed by a type, a list, and a mapping function that assigns types to the index list. Heterogeneous lists are of the same length as their index list, and the type of their $i$th element should be the result of applying the mapping function to the $i$th element of the index list. For example, supposing that $A$ is the boolean type in the metatheory, and the mapping function $F$ assigns $\mathbb{N}$ to \textit{true} and $\app{\texttt{List}}{\mathbb{N}}$ to \textit{false}, the following example is a heterogeneous list with the specified type:
\begin{equation*}
	\app{\textit{hcons}}{100,(\app{\textit{hcons}}{\seplist123*,\textit{hnil}})} : \app{\texttt{HList}}{\mathbb{B},F,\seplist{\textit{true}}{\textit{false}}*}
\end{equation*}

We leave the parameters $x$ and $xs$ (of \textit{hcons}) implicit, since they are inferable from the tails and the type of the list. For readability, we use the notation $\sephetlist{e_1, \ldots, e_n}$ to denote finite heterogeneous lists. Furthermore, we use $\hlistmap{f}{\textit{xs}}$ to denote the application of function $f$ to all elements of the heterogeneous list $\textit{xs}$. We omit the formal definition of this function, and rather refer to~\cite{kurucz2025on} for further details.

Finally, we also utilize record syntax, which allows us to refer to parameters of a (single-constructor) inductive type's constructor as fields in the following way:

\medskip\hfill
\begin{minipage}{0.3\textwidth}
\typedef[record]{\texttt{RecType}}[\dict{Type}]{
	{\textit{field}_1}{A},
	{\ldots}{\ldots},
	{\textit{field}_n}{Z}
}
\end{minipage}
\begin{minipage}{0.49\textwidth}
\typedef[inductive]{\texttt{RecType}}[\dict{Type}]{
	{\textit{mkRecType}}{A \to \ldots \to Z \to \texttt{RecType}}
}
\end{minipage}\hfill
\medskip

With record syntax, we apply field names to records to access the parameters of the constructor, for example, if $\textit{rec} : \texttt{RecType}$, then $\app {\textit{field}_1} {\textit{rec}}$ denotes the first field of $\textit{rec}$, which is of type $A$.
\end{addtoappendix}

Matching logic~\cite{chen2019mu} provides the theoretical foundation for the $\mathbb{K}$  framework. While there are multiple variants of this logic (e.g., applicative matching logic~\cite{chen2021explained}), we present our results based on our previous work~\cite{kurucz2025on} (utilizing a variant of matching $\mu$-logic~\cite{chen2019mu}), and use locally nameless representation~\cite{chargueraud2011locally} which closely follows in the footsteps of the machine-checked implementation~\cite{kore-ml}. This choice eliminates the need to reason about $\alpha$-equivalence, and makes it simpler to mechanize capture-avoiding substitutions.

First, we summarize the syntax of the logic, starting with the definition of signatures.

\begin{definition}[Signature]\label{def:sig}
	A matching logic signature is a tuple $(\sepcomma {\dict{Sorts}} {\EVar} {\SVar} {\Sigma} *)$, where
	\begin{itemize}
		\item $\dict{Sorts}$ is a set of sorts (denoted with $\sepcomma {s} {s'} {\ldots} *$);
		\item $\EVar = \{\EVar[s]\}_{s \in \dict{Sorts}}$ is a countably infinite set of element variables, indexed by sorts such as $s$ (denoted with $\sepcomma {\ofsortvar{x}{s}} {\ofsortvar{y}{s}} {\ldots} *$);
		\item $\SVar = \{\SVar[s]\}_{s \in \dict{Sorts}} $ is a countably infinite set of set variables, indexed by sorts such as $s$ (denoted with $\sepcomma {\ofsortvar{X}{s}} {\ofsortvar{Y}{s}} {\ldots} *$);
		\item $\Sigma = \{\Sigma_{\sepcomma {s_1} {\dots} {s_n} s *}\}_{\sepcomma {s_1} {\dots} {s_n} s * \in \dict{Sorts}}$, where $s_i$ are the parameter sorts and $s$ is the return sort, is a countable set of many-sorted symbols (denoted with $\sepcomma {\sigma} {\sigma'} {\ldots} *$).
	\end{itemize}

	When the details are not relevant, we use $\Sigma$ to refer to the entire signature.
\end{definition}

Next, we present the syntax of matching logic formulas (so-called \emph{patterns}) which fill the role of both terms and formulas of first-order logic.

\begin{definition}[Pattern]\label{def:pattern}
	Given a signature $\Sigma$, we define patterns (denoted by $\sepcomma {\varphi_s} {\psi_s} *$, where the subscript indicates the sort of the pattern) with the following syntax:
	\begin{equation*}
		\ofsort{\varphi}{s} \coloneqq
		\ofsortvar{\fevar x}{s} \mid
		\ofsortvar{\fsvar X}{s} \mid
		\ofsortvar{\bevar n}{s} \mid
		\ofsortvar{\bsvar N}{s} \mid
		\ofsort{\lnot\varphi}{s} \mid
		\ofsort{\varphi}{s} \land \ofsort{\varphi'}{s} \mid
		\exists_{s'} \ld \ofsort{\varphi}{s} \mid
		\mu \ld \ofsort{\varphi}{s} \mid
		\polyapp {\sigma} {\varphi_{s_1}} {\ldots} {\varphi_{s_n}} * \quad \text{if } \sigma \in \Sigma_{\sepcomma {s_1} {\dots} {s_n} s *}
	\end{equation*}

	These are: free element variable, free set variable, bound element variable, bound set variable, negation, conjunction, existential quantifier, least fixed point quantifier, and pattern application. Negation binds the tightest, then conjunction. The scope of quantifiers extends as far to the right as possible. The sort $s'$ after the existential quantifier denotes the sort of the newly introduced bound variable.
\end{definition}

We omit the sorts from variables, where it is inferable. As mentioned before, we utilize a locally nameless representation, which uses names for free variables and de Bruijn indices for bound variables. These de Bruijn indices indicate the exact number of nested quantifiers ($\exists$ for element variables, $\mu$ for set variables) between the variable's occurrence and the specific quantifier that binds it. As an example, supposing that we have two symbols \textit{isZero} and \textit{plus}, and sorts \textit{nat} and \textit{bool}, we present a pattern expressing that there are two numbers whose sum is zero with both named and (locally) nameless variables.
\begin{equation}\label{ex:lnpat}
    \exists \ofsortvar{x}{\textit{nat}} \ld \polyapp {\textit{isZero}} {\exists \ofsortvar{y}{\textit{nat}} \ld \polyapp{\textit{plus}} {\fevar y} {\fevar x} *} * 
    \equiv
    \exists_{\textit{nat}} \ld \polyapp {\textit{isZero}} {\exists_{\textit{nat}} \ld \polyapp{\textit{plus}} {\bevar 0} {\bevar 1} *} *
\end{equation}

We call de Bruijn indices dangling if they do not have a corresponding binder. We use $\FV[s']{\ofsort{\varphi}{s}}$ to denote the free variables of sort $s'$ of pattern $\ofsort{\varphi}{s}$. Furthermore, we use $\Gamma$ to denote a set of matching logic patterns (we also call this set a theory). Next, we introduce capture avoiding substitutions (for any kind of variable), and present an example to demonstrate it in the locally nameless setting.

\begin{definition}[Substitution]\label{def:sub}
	We denote capture-avoiding, many-sorted substitutions as $\subst{\ofsort{\varphi}{s}}{\ofsortvar{x}{s'}}{\ofsort{\psi}{s'}}$, meaning a replacement of every occurrence of $x$ in $\ofsort{\varphi}{s}$ with $\ofsort{\psi}{s'}$. For the complete definition, we refer to the formalization~\cite{kore-ml}. Note that the variable $x$ may be any of the four types of variables (free or bound, element or set variable). In the case of bound variables, only dangling ones can be substituted, and we take into account the quantifiers' ability to increase the index. E.g., $\subst{(\exists_s \ld \ofsortvar{\bevar 0}{s} \land \ofsortvar{\bevar 1}{s})}{\ofsortvar{\bevar 0}{s}}{\psi} = \exists_s \ld \subst{(\ofsortvar{\bevar 0}{s} \land \ofsortvar{\bevar 1}{s})}{\ofsortvar{\bevar 1}{s}}{\psi} = \exists_s \ld \subst{(\ofsortvar{\bevar 0}{s})}{\ofsortvar{\bevar 1}{s}}{\psi} \land \subst{(\ofsortvar{\bevar 1}{s})}{\ofsortvar{\bevar 1}{s}}{\psi} = \exists_s \ld \ofsortvar{\bevar 0}{s} \land \psi$.
\end{definition}

After revisiting the syntax of matching logic, we also outline models and the semantics.

\begin{definition}[Model]\label{def-model}
	A ($\Sigma$-)model is a tuple $M = (\sepcomma {\{M_s\}_{s \in \dict{Sorts}}} {\{\sigma_M\}_{\sigma \in \Sigma}} *)$, where
\begin{itemize}
	\item $M_s$ is the non-empty carrier set (domain) for each sort; 
	\item $\sigma_M : M_{s_1} \times \dots \times M_{s_n} \to \mathcal{P}(M_s)$ serves as the interpretation of the symbol $\sigma \in \Sigma_{\sepcomma {s_1} {\dots} {s_n} s *}$.
\end{itemize}
\end{definition}

Symbol interpretation can be extended to handle model subsets in the following way.

\begin{equation}\label{eq-appext}\tag{\textsc{AppExt}}
	\polyapp {\sigma_M} {A_1} {\dots} {A_n} * \coloneqq \bigcup_{a_i \in A_i} \polyapp {\sigma_M} {a_1} {\dots} {a_n} *
\end{equation}

The meaning of variables is expressed in terms of \emph{valuations}. Element variables are interpreted as model elements, while set variables are interpreted as model subsets.

\begin{definition}[Valuation of variables]\label{def-valuations}
A valuation $\rho$ is a pair of valuations for element and set variables for each $s \in \dict{Sorts}$, that is, $\rho : (\EVar[s] \to M_s) \times (\SVar[s] \to \mathcal{P}(M_s))$. We use $\update{\rho}{x}{m}$ to update valuation $\rho$ by mapping the element variable $\ofsortvar{x}{s}$ to the carrier \emph{element} $m \in M_s$. Respectively, we use $\update{\rho}{X}{A}$ to update the valuation of set variables of $\rho$, by mapping the set variable $\ofsortvar{X}{s}$ to the carrier \emph{subset} $A \subseteq M_s$.
\end{definition}

With these concepts, we recall the semantics of matching logic expressed in terms of set operations.

\begin{definition}[Interpretation of patterns]\label{def:sem} Given a matching logic signature $\Sigma$, model $M$, and a valuation $\rho$, we define pattern interpretation (mapping patterns of sort $s$ to subsets of $M_s$) the following way:

\begin{gather*}
	\bar \rho (\ofsortvar{\fevar x}{s}) \coloneqq \{\rho (x)\} \qquad\qquad 
	\bar \rho (\ofsortvar{\fsvar X}{s}) \coloneqq \rho (X) \qquad\qquad 
	\bar \rho (\ofsortvar{\bevar n}{s}) \coloneqq \emptyset \qquad\qquad 
	\bar \rho (\ofsortvar{\bsvar N}{s}) \coloneqq \emptyset \\
	\bar \rho (\ofsort{\varphi}{s} \land \ofsort{\varphi'}{s}) \coloneqq \bar \rho (\ofsort{\varphi}{s}) \cap \bar \rho (\ofsort{\varphi'}{s}) \qquad\qquad
	\bar \rho (\lnot \ofsort{\varphi}{s}) \coloneqq M_s \setminus \bar \rho (\ofsort{\varphi}{s}) \\
	\bar \rho (\exists_{s'} \ld \ofsort{\varphi}{s}) \coloneqq \bigcup_{m \in M_{s'}} \overline{\update{\rho}{x}{m}} (\subst{\ofsort{\varphi}{s}}{\ofsortvar{\bevar 0}{s'}}{\ofsortvar{x}{s'}}), \text{where } x \notin \FV[s']{\ofsort{\varphi}{s}} \\
  \bar \rho (\mu \ld \ofsort{\varphi}{s}) \coloneqq \lfp \mathcal{F}^\rho_{\ofsort{\varphi}{s}, \ofsortvar{X}{s}}, \text{where } \mathcal{F}^\rho_{\ofsort{\varphi}{s}, \ofsortvar{X}{s}}(A) = \overline{\update{\rho}{X}{A}} (\subst{\ofsort{\varphi}{s}}{\ofsortvar{\bsvar 0}{s}}{\ofsortvar{X}{s}}) \text{ and } X \notin \FV[s]{\ofsort{\varphi}{s}}\\
	\bar \rho (\polyapp {\sigma} {\varphi_{s_1}} {\dots} {\varphi_{s_n}} *) \coloneqq \polyapp {\sigma_M} {\bar\rho(\varphi_{s_1})} {\dots} {\bar\rho(\varphi_{s_n})} *, \text{where } \sigma \in \Sigma_{\sepcomma {s_1} {\dots} {s_n} s *}
\end{gather*}
We use $\lfp \mathcal{F}$ to denote the least fixpoint of a monotone function $\mathcal{F}$, and define it as an intersection of pre-fixpoints as in~\cite{chen2019mu}.\footnote{Note that $\mathcal{F}^\rho_{\ofsort{\varphi}{s}}$ is only monotone when $\ofsort{\varphi}{s}$ is a positive pattern (as defined in~\cite{chen2019mu}), which is crucial when proving the soundness of a logical deduction system. However, in this paper, we do not define syntactic proofs; therefore, positivity of patterns is not relevant to our definitions.} To define a total semantics for the locally nameless representation, we include an empty interpretation for dangling bound variables; however, only closed patterns should be evaluated.
\end{definition}

Finally, we revisit the definition of validity and satisfaction in matching logic.

\begin{definition}[Satisfaction]\label{def:sat}
A pattern $\ofsort{\varphi}{s}$ holds in a model $M$ (denoted as $M \models \ofsort{\varphi}{s}$), if for all valuations $\rho$, $\bar \rho(\ofsort{\varphi}{s}) = M_s$. A model $M$ satisfies (or validates) a theory $\Gamma$ ($M \models \Gamma$), if all patterns of $\Gamma$ hold in $M$.
\end{definition}

\paragraph{A dependently typed implementation of matching logic}\label{sec:dependent}

The matching logic variant discussed here has also been implemented in the \Coq{} proof assistant~\cite{kurucz2025on,kore-ml}. Next, we review the necessary abstractions from this Rocq encoding to support later sections, and for more details, we refer to~\cite{kurucz2025on}. We do not recall well-understood inductive types namely, \dict{List}, \texttt{Option}, union (\texttt{Sum}), and dependent pairs (\texttt{SigT}) but refer to \Cref{sec:depintro} for their formal definition. We use the following notations for these types:
\begin{itemize}
    \item For \dict{List}, $[]$ denotes the empty list, $x \conscc \textit{xs}$ prepends $x$ to the list $\textit{xs}$, and $\map{f}{\textit{xs}}$ applies the function $f$ to all elements of $\textit{xs}$.
    \item For \texttt{Option}, \dict{None} denotes the empty, and \dict{Some} denotes the non-empty constructor.
    \item For \texttt{Sum}, \dict{inl} and \dict{inr} are the constructors that inject a value from the left or right type into the union.
    \item Dependent pairs (\texttt{SigT}) are denoted by $\DPair{A}{x}{P}$ (where $P$ can depend on $x$), while $\fst{p}$ and $\snd{p}$ is used to obtain the first and second  element of a (dependent) pair $p$.
\end{itemize}
First, we present dependently typed matching logic signatures, and two auxiliary inductive types (heterogeneous lists and list membership) used in the encoding of patterns.

\medskip
\noindent\hfill\begin{minipage}{0.35\textwidth}
\typedef[record]{\texttt{Signature}}[\dict{Type}]{
    {\dict{Sorts}}{\dict{Type}},
    {\dict{EV}}{\dict{Sort} \to \dict{Type}},
    {\dict{SV}}{\dict{Sort} \to \dict{Type}},
    {\dict{symbols}}{\dict{Type}},
    {\dict{params}}{\dict{symbols} \to \listof{\dict{Sorts}}},
    {\dict{return}}{\dict{symbols} \to \dict{Sorts}}
}
\end{minipage}\hfill\begin{minipage}{0.60\textwidth}
\typedef[inductive]{\app{\texttt{HList}}{A, (F : A \to \dict{Type})}}[\listof{A} \to \dict{Type}]{
    {\textit{hnil}}{\app {\texttt{HList}}{A, F, []}},
    {\textit{hcons}}{\fa {x, \textit{xs}} \app{F}{x} \to \app {\texttt{HList}}{A, F, \textit{xs}} \to \app {\texttt{HList}}{A, F, (x \conscc \textit{xs})}}
}

\typedef{\app{\dict{In}}{A}}[A \to \listof{A} \to \dict{Type}]{
    {\dict{ineq}}{\fa {x, xs} \app{\dict{In}}{A, x, (x \conscc xs)}},
    {\dict{incons}}{\fa {x, y, xs} \app{\dict{In}}{A, y, xs} \to \app{\dict{In}}{A, y, (x \conscc xs)}}
}
\end{minipage}\hfill
\medskip

A signature (\Cref{def:sig}) is expressed as a record (i.e. a single constructor inductive type with named constructor parameters) containing the variables and symbols. We use two functions \dict{params} and \dict{return} to encode the arity of the symbols. Moreover, in the implementation~\cite{kore-ml}, two further fields ensure the infiniteness of variable names, but we omit these here for simplicity.

Heterogeneous lists are indexed by a normal list and a function. The type of the $i$th element in the heterogeneous list is determined by applying the indexing function to the $i$th element of the index list. For example, $\app{\textit{hcons}}{100,(\app{\textit{hcons}}{\seplist123*,\textit{hnil}})}$ is of type $\app{\texttt{HList}}{\mathbb{B},F,\seplist{\textit{true}}{\textit{false}}*}$, where the index function $F$ maps \textit{true} to $\mathbb{N}$ and \textit{false} to $\listof{\mathbb{N}}$.

The two constructors of \texttt{In} establish membership within a non-empty list by asserting that an element is either the head or resides within the tail. We also present two lemmas relating list membership and mapping (the second one only holds for injective functions), which is leveraged in the following sections.

\begin{lemma}
\label{lem:inmap}
    $\dict{Inmap} : \fa {A, B, (f : A \to B), (x : A), (\textit{xs} : {\normalfont\texttt{List}}\ A)} \app{\normalfont\dict{In}}{A, x, xs} \to \app{\normalfont\dict{In}}{B, (\app{f}{x}), (\map{f}{xs})}$
\end{lemma}

\begin{lemma}
\label{lem:inmapr}
    $\dict{Inmapr} : \fa {A, B, (f : A \to B), (x : A), (\textit{xs} : {\normalfont\texttt{List}}\ A)} (\fa {a, b} \app fa = \app fb \to a = b) \to \\\app{\normalfont\dict{In}}{B, (\app{f}{x}), (\map{f}{xs})} \to \app{\normalfont\dict{In}}{A, x, xs} $
\end{lemma}

Elements of type \texttt{In} represent more than structural proofs of membership; they additionally serve as indices: the number of nested \dict{incons} constructors corresponds directly to the included element's precise index. For this reason, we utilize them as de Bruijn indices in the following encoding of matching logic patterns. We use $\placeholder$ to denote the arguments listed in the constructors (rather than using just constructor names), except \textit{ex}, \textit{mu} and \textit{s}, which are inferrable from the context. For example, in the case of the existential quantifier, $s'$ and the unnamed pattern argument must be placed in the $\placeholder$s. Calling the latter $\varphi$, we obtain $\exists_{s'} \ld \varphi$, matching the syntax presented in \Cref{def:pattern}.

\medskip
\noindent
\typedef[inductive]{\app{\texttt{Pattern}}{(\Sigma : \texttt{Signature})}}[\listof{\dict{Sorts}} \to \listof{\dict{Sorts}} \to \dict{Sorts} \to \dict{Type}]{
		{\fevar \placeholder}{\fa {s, ex, mu} (\app {\EVar}{s}) \to \sortedclosedpattern{s}{ex}{mu}},
		{\fsvar \placeholder}{\fa {s, ex, mu} (\app {\SVar}{s}) \to \sortedclosedpattern{s}{ex}{mu}},
		{\bevar \placeholder}{\fa {s, ex, mu} \app {\texttt{In}}{s, ex}\to \sortedclosedpattern{s}{ex}{mu}},
		{\bsvar \placeholder}{\fa {s, ex, mu} \app {\texttt{In}}{s, mu} \to \sortedclosedpattern{s}{ex}{mu}},
		{\syapp \placeholder \placeholder}{\fa {ex, mu, (\sigma : \Sigma)} \app {\texttt{HList}}{\dict{Sorts}, (\closedpattern{ex}{mu}), (\app {\argsorts} {\sigma})} \to \sortedclosedpattern{(\app {\retsorts} {\sigma})}{ex}{mu}},
		{\lnot \placeholder}{\fa {s, ex, mu} \sortedclosedpattern{s}{ex}{mu} \to \sortedclosedpattern{s}{ex}{mu}},
		{\placeholder \land \placeholder}{\fa {s, ex, mu} \sortedclosedpattern{s}{ex}{mu} \to \sortedclosedpattern{s}{ex}{mu} \to \sortedclosedpattern{s}{ex}{mu}},
		{\exists_{\placeholder} \ld \placeholder}{\fa {s, (s' : \dict{Sorts}), ex, mu} \sortedclosedpattern{s}{(s' \conscc ex)}{mu} \to \sortedclosedpattern{s}{ex}{mu}},
		{\mu \ld \placeholder}{\fa {s, ex, mu} \sortedclosedpattern{s}{ex}{(s \conscc mu)} \to \sortedclosedpattern{s}{ex}{mu}}
}
\medskip

The specialty of this encoding is that scoping and sorting is built into the type indices; therefore, ill-sorted or ill-scoped patterns cannot be constructed.\footnote{We note that the syntactic condition for positivity in case of $\mu$ binders is not encoded into the syntax currently.} The implicit arguments $\textit{ex}$ and $ \textit{mu}$ express two lists (environments) containing the sorts of the available de Bruijn indices (bound variables). Quantifiers extend the environments for their bodies, while de Bruijn indices are represented as values of \texttt{In}, as discussed above.  For readability, we keep using natural numbers to represent de Bruijn indices, but note that these numbers directly correspond to the number of \dict{incons} constructors in a value of type \texttt{In}.

We also draw attention to the syntax of application patterns ($\syapp \placeholder \placeholder$). Since its arguments can have different sorts, we utilized heterogeneous lists to express these argument patterns, specified for a symbol $\sigma$. Revisiting the example in \Cref{ex:lnpat}, the following subpattern has the following type.\footnote{We use angled parentheses for the arguments instead of round ones as in Def.~\ref{def:pattern} to indicate the use of heterogeneous lists.}

\begin{equation*}
    \syapp {\textit{isZero}} {\sephetlist{\exists_{\textit{nat}} \ld \syapp{\textit{plus}} {\sephetlist {\bevar 0, \bevar 1}}}} : \sortedclosedpattern{\textit{bool}}{(\textit{nat} \conscc \textit{ex})}{\textit{mu}}
\end{equation*}

In this example, \textit{ex} and \textit{mu} can be any sorting environments (potentially empty), but it is important that the dangling index $\bevar 1$ (which would be $\bevar 0$ outside of the binder) needs to be of sort \textit{nat}; therefore, \textit{nat} is the first element of the environment for element variables. Next, we proceed with presenting the definition of dependently typed models and variable valuations.

\medskip\hfill
\begin{minipage}{0.52\textwidth}
\noindent
\typedef[record]{\app{\texttt{Model}}{(\Sigma : \texttt{Signature})}}[\dict{Type}]{
    {\textit{carrier}}{\dict{Sorts} \to \dict{Type}},
    {\textit{inh}}{\fa {s} \app{\textit{carrier}}{s}},
    {\textit{app}}{\fa {\sigma} \app {\texttt{HList}} {\dict{Sorts}, \textit{carrier}, (\app {\argsorts}{\sigma})} \to\\& \mathcal P(\app {\textit{carrier}}{(\app {\retsorts} {\sigma})})}
}
\end{minipage}\hfill
\begin{minipage}{0.44\textwidth}
\noindent
\typedef[record]{\app{\texttt{Valuation}}{\Sigma, (M : \app{\texttt{Model}}{\Sigma})}}[Type]{
    {\textit{evar\_val}}{\fa{s} \app{\EVar}{\Sigma, s} \to \app {\textit{carrier}}{M, s}},
    {\textit{svar\_val}}{\fa{s} \app{\SVar}{\Sigma, s} \to \mathcal P(\app {\textit{carrier}}{M, s})}
}
\end{minipage}\hfill
\medskip

In the definition of models, we use the same idea as before: heterogeneous lists are used to represent carrier elements for the arguments of polyadic, many-sorted symbol interpretation. The second field (\textit{inh}) expresses that all the carriers for the sorts should have at least one element.

Variable valuations are defined as expected: using a pair of functions that, for each sort, assign an element to element variables, and a subset of the carrier to set variables respectively.

Finally, we highlight some branches of the machine-checked semantics of matching logic: the semantics of element variables, conjunction, application, and existential quantification. We do not recall the machine-checked definition of extended symbol application (denoted by \textit{app\_ext}, defined in \Cref{eq-appext}), heterogeneous list mapping (denoted by $\hlistmap{}{}$), valuation update (\Cref{def-valuations}), bound variable substitution (denoted by \textit{bevar\_subst}, introduced in \Cref{def:sub}), and how fresh variables are produced via $\textit{fresh}$, rather refer to~\cite{kurucz2025on,kore-ml}. We also omit the \texttt{Signature} and \texttt{Model} parameters from the \texttt{Model}, \texttt{Valuation}, and \texttt{Pattern} types.
\fundef{\dict{eval}}{\fa {\Sigma, ex, mu, s, M} \texttt{Valuation} \to \sortedclosedpattern{s}{ex}{mu} \to \mathcal P(\app {\textit{carrier}}{M, s})}{
 	{\rho\ (\ofsortvar{\bevar x}{s})}{\emptyset},
    {\rho\ (\ofsortvar{\fevar x}{s})}{\{\app{\textit{evar\_val}}{\rho, x}\}},
    {\rho\ (\varphi_1 \land \varphi_2)}{\app {\textit{eval}} {\rho, \varphi_1} \cap \app {\textit{eval}} {\rho, \varphi_2}},
    {\rho\ (\syapp{\sigma}{xs})}{\app {\textit{app\_ext}} {\sigma, (\hlistmap{\app {\textit{eval}} {\rho}}{xs})}},
    {\rho\ (\exists_{s'} \ld \varphi)}{\bigcup_{\mathclap{c : \app {\textit{carrier}}{M, s'}}} \app {\textit{eval}}{(\update{\rho}{x}{c}),(\app {\textit{bevar\_subst}} {(ex \coloneqq []), \fevar x, \varphi})} \qquad\scalebox{.85}{where $x = \app {\textit{fresh}_{s'}} {(\FV[s']{\varphi})}$}}
}

In the machine-checked definition, we rely on a variant of set theory implemented in the \Coq{} library called \emph{stdpp}, which provides the necessary set operations. The idea is to represent a set containing $A$-typed elements based on a judgment $A \to \dict{Prop}$ (which is essentially a subtype of \dict{Type} in \Coq{}). The semantics should not be used to evaluate ill-formed patterns (which contain dangling variables), but to construct a total definition, we assign the empty set to dangling variables. In case of binders, correctly bound variables are replaced by free variables (by \textit{bevar\_subst}, where $\textit{ex} := []$ specifies that the outermost de Bruijn index needs to be replaced). Free variables are evaluated based on the valuation functions (which is updated for the fresh variables introduced during the evaluation of binders). Conjunction is mapped to set intersection, application utilizes heterogeneous list mapping (formally defined in~\cite{kore-ml,kurucz2025on}), and the extended interpretation for symbol application.

Finally, we encode dependently typed theories as $\mathcal P(\DPair{\dict{Sorts}}{s}{\sortedclosedpattern{s}{[]}{[]}})$. It is impossible to represent theories as a simple set of patterns because patterns are sorted; thus, their type is indexed by their sorts; therefore, patterns of distinct sorts have different types. Instead, we use dependent pairs within the theory which include a sort and a closed pattern (i.e. its environments are empty) of that sort.

\section{Combining models}\label{operators}

There are two main operators for reusing existing models: gluing and extension. Here, we explore possible definitions and properties of the former, and briefly discuss how to adapt it to encode the latter.

It is reasonable to expect that each model is defined over its own signature containing different sorts and symbols. Therefore, we begin by defining these operations over signatures. Gluing is a balanced operation, combining two structures by merging all definitions they provide. In contrast, extension is an unbalanced version that extends a structure only with certain elements (e.g., signatures with symbols or models and theories with axioms). We expect that, given two models with different signatures and a theory for the models, each validated by its respective model, it is automatically derivable that the composite model resulting from either operation validates the concatenation of the two theories.

\paragraph{Equivalences.}

We define isomorphisms to identify types that cannot be proven identical, but contain the same information. This is expressed using two functions that convert between the values of each type, and two proofs that these functions are inverses of each other in both directions.

A similar relation is that of partial isomorphisms, which holds when one type contains potentially less information than the other. We express this by altering one of the functions to have an optional result, and adjust the inverse restrictions to accommodate this: in one direction, a value must always be returned, while the other only needs to hold if one is returned. We express these concepts in type theory as records:

\medskip\hfill%
\begin{minipage}{.48\linewidth}
    \noindent
    \typedef[record]{\app{\dict{iso}}{(A, B : \dict{Type})}}[\dict{Type}]{
    	{\dict{partialf}'}{A \to B},
    	{\dict{partialg}'}{B \to A},
    	{\dict{partialgf}'}{\fa x \app{\dict{partialg}'}{(\app{\dict{partialf}'}{x})} = x},
    	{\dict{partialfg}'}{\fa x \app{\dict{partialf}'}{(\app{\dict{partialg}'}{x})} = x}
    }
\end{minipage}%
\begin{minipage}{.48\linewidth}
    \noindent
    \typedef[record]{\app{\dict{partial}}{(A, B : \dict{Type})}}[\dict{Type}]{
    	{\dict{partialf}}{A \to B},
    	{\dict{partialg}}{B \to \app{\dict{option}}{A}},
    	{\dict{partialgf}}{\fa x \app{\dict{partialg}}{(\app{\dict{partialf}}{x})} = \app{\dict{Some}}{x}},
    	{\dict{partialfg}}{\fa {x, y} \app{\dict{partialg}}{x} = \app{\dict{Some}}{y} \to \app{\dict{partialf}}{y} = x}
    }
\end{minipage}%
\hfill\medskip

\dict{iso} is an equivalence relation, while \dict{partial} is a partial order up to isomorphism. A well-known mathematical example of isomorphic types are natural numbers and integers. Partial isomorhism holds between any finite type and natural numbers, for example.

We can also define the equivalence of \textbf{record} types using the pairwise equivalence of relevant fields. For example, we may consider two signatures (shown in \Cref{sec:dependent}) equivalent, if their sort and symbol types are isomorphic, both variables return isomorphic types for appropriately converted (via $\dict{partialf}'$ and $\dict{partialg}'$) sorts, while likewise-convertible symbols have convertible parameter and return sorts. In contrast, if the exact value of a field is irrelevant, it only needs to be definable, it can be omitted from the comparison. An example of this is the \dict{inh} field of models. We denote these equivalence relations, including $\dict{iso}$, with $\cong$, while partial isomorphisms are always denoted explicitly by \dict{partial}.

\subsection{Combination of signatures}

The simplest approach to combining signatures is taking the disjoint union of the types defined in the signature, sorts and symbols. This ensures that these same types in the combined signature not only contain every sort and symbol from the inputs, but also that they contain no new, ``junk'' values that cannot be handled. This way, the union-based combination provides a simple way to identify which signature each object originates from via pattern matching, and thereby to delegate the other operations over them to the same method of the appropriate signature.

This approach has very few restrictions and is useful when the models describe entirely distinct concepts, such as those serving as a full specification of a commonly used type. However, in more complex models that may already contain several sorts that would need to be identified in their union, the sum-based approach becomes problematic. It is possible to assign the same carrier set to them; however, complete equality cannot be stated. If the logic is extended with a subsorting relation (as in~\cite{kurucz2025on}), its antisymmetry property can be exploited for this purpose. By allowing the user to mark identical sorts as subsorts of one another, they will also become equal by antisymmetry.

This sum also highlights certain properties of this combination that a generalization must maintain. We know that disjoint unions are commutative and associative under isomorphism: $A \uplus B \cong B \uplus A$ and $(A \uplus B) \uplus C \cong A \uplus (B \uplus C)$. Any form of signature combination must also fulfill these properties, as the order in which they are merged should not matter.

We use this idea to serve as a blueprint for defining this combination in a generalized way. New types representing sorts and symbols must be defined in a way that they incorporate these types from the constituent signatures, which can be ensured using partial isomorphisms. Furthermore, these types shall not contain additional values, which is enforceable via another restriction over both isomorphisms, stating that at least one of them must be able to convert back. However, it is not strictly necessary to ensure that only one does so, removing the disjoint nature of the union. This allows identifying sorts at a deeper level, without the need to rely on subsorting. The functions of the signatures are required to agree on their result, if both can handle the given input.

\begin{definition}[Conditions for signature combining]\label{def:sigcombinecond} We define combinable signatures as having two types capable of representing the strict union of their sorts and symbols, and agreeing on the result of their functions when applied to shared sorts or symbols.

\medskip\noindent
\typedef[class]{\app{\dict{gluable}}{(\Sigma_1, \Sigma_2 : \texttt{Signature})}}{
	{\dict{SharedSort}}{\dict{Type}},
	{\dict{subone}}{\app{\dict{partial}}{(\app{\dict{Sorts}}{\Sigma_1}), \dict{SharedSort}}},
	{\dict{subtwo}}{\app{\dict{partial}}{(\app{\dict{Sorts}}{\Sigma_2}), \dict{SharedSort}}},
	{\dict{subtotalone}}{\fa s (\te {s'} \app{\dict{partialg}}{\dict{subone}, s} = \app{\dict{Some}}{s'}) \lor \te {s'} \app{\dict{partialg}}{\dict{subtwo}, s} = \app{\dict{Some}}{s'}},
    {\dict{agreeone}}{\fa {s, s_1, s_2} \app{\dict{partialg}}{\dict{subone}, s} = \app{\dict{Some}}{s_1} \to \app{\dict{partialg}}{\dict{subtwo}, s} = \app{\dict{Some}}{s_2} \to \app{\dict{EV}}{\Sigma_1, s_1} \cong \app{\dict{EV}}{\Sigma_2, s_2}},
    {\dict{agreetwo}}{\fa {s, s_1, s_2} \app{\dict{partialg}}{\dict{subone}, s} = \app{\dict{Some}}{s_1} \to \app{\dict{partialg}}{\dict{subtwo}, s} = \app{\dict{Some}}{s_2} \to \app{\dict{SV}}{\Sigma_1, s_1} \cong \app{\dict{SV}}{\Sigma_2, s_2}},
	{\dict{SharedSymbols}}{\dict{Type}},
	{\dict{subfive}}{\app{\dict{partial}}{(\app{\dict{symbols}}{\Sigma_1}), \dict{SharedSymbols}}},
	{\dict{subsix}}{\app{\dict{partial}}{(\app{\dict{symbols}}{\Sigma_2}), \dict{SharedSymbols}}},
	{\dict{subtotalthree}}{\fa \sigma (\te {\sigma'} \app{\dict{partialg}}{\dict{subfive}, \sigma} = \app{\dict{Some}}{\sigma'}) \lor \te {\sigma'} \app{\dict{partialg}}{\dict{subsix}, \sigma} = \app{\dict{Some}}{\sigma'}},
    {\dict{agreethree}}{\fa {\sigma, \sigma_1, \sigma_2} \app{\dict{partialg}}{\dict{subfive}, \sigma} = \app{\dict{Some}}{\sigma_1} \to \app{\dict{partialg}}{\dict{subsix}, \sigma} = \app{\dict{Some}}{\sigma_2} \to \breakrow\map{(\app{\dict{partialf}}{\dict{subone}})}{(\app{\dict{params}}{\Sigma_1, \sigma_1})} = \map{(\app{\dict{partialf}}{\dict{subtwo}})}{(\app{\dict{params}}{\Sigma_2, \sigma_2})}},
    {\dict{agreefour}}{\fa {\sigma, \sigma_1, \sigma_2} \app{\dict{partialg}}{\dict{subfive}, \sigma} = \app{\dict{Some}}{\sigma_1} \to \app{\dict{partialg}}{\dict{subsix}, \sigma} = \app{\dict{Some}}{\sigma_2} \to \breakrow\app{\dict{partialf}}{\dict{subone}, (\app{\dict{return}}{\Sigma_1, \sigma_1})} = \app{\dict{partialf}}{\dict{subtwo}, (\app{\dict{return}}{\Sigma_2, \sigma_2})}}
}
\end{definition}

Next, we provide an instance of \dict{partial} for both sides of the union of arbitrary types $A$ and $B$, followed by an instance of \dict{gluable} based on the union type and these two partial instances.

\hfill%
\begin{minipage}{.48\linewidth}
    \noindent\fundef*{\dict{inlpartial}}{\fa {A, B} \app{\dict{partial}}{A, (A \uplus B)}}{
    	{\dict{partialf}}{\unskip}{\unskip}{\dict{inl}},
    	{\dict{partialg}}{\unskip}{x}{\othercase{\app{\dict{Some}}{x'}}{x = \app{\dict{inl}}{x'}}{\dict{None}}},
    	{\ldots}{\unskip}{\unskip}{\ldots}
    }
\end{minipage}%
\begin{minipage}{.48\linewidth}
    \noindent\fundef*{\dict{inrpartial}}{\fa {A, B} \app{\dict{partial}}{B, (A \uplus B)}}{
    	{\dict{partialf}}{\unskip}{\unskip}{\dict{inr}},
    	{\dict{partialg}}{\unskip}{x}{\othercase{\app{\dict{Some}}{x'}}{x = \app{\dict{inr}}{x'}}{\dict{None}}},
    	{\ldots}{\unskip}{\unskip}{\ldots}
    }
\end{minipage}%
\hfill

The remaining conditions are proven using case separation. The \dict{gluable} instance is as follows.
\fundef*!{\dict{defaultsum}}{\fa {\Sigma_1, \Sigma_2} \app{\dict{gluable}}{\Sigma_1, \Sigma_2}}{
	{\dict{SharedSort}}{\unskip}{\unskip}{\app{\dict{Sorts}}{\Sigma_1} \uplus \;\app{\dict{Sorts}}{\Sigma_2}\setnobreak},
	{\dict{SharedSymbols}}{\unskip}{\unskip}{\app{\dict{symbols}}{\Sigma_1} \uplus \;\app{\dict{symbols}}{\Sigma_2}},
	{\dict{subone}}{\unskip}{\unskip}{\app{\dict{inlpartial}}{(\app{\dict{Sorts}}{\Sigma_1}), (\app{\dict{Sorts}}{\Sigma_2})}\setnobreak},
	{\dict{subtwo}}{\unskip}{\unskip}{\app{\dict{inrpartial}}{(\app{\dict{Sorts}}{\Sigma_1}), (\app{\dict{Sorts}}{\Sigma_2})}},
	{\dict{subfive}}{\unskip}{\unskip}{\app{\dict{inlpartial}}{(\app{\dict{symbols}}{\Sigma_1}), (\app{\dict{symbols}}{\Sigma_2})}\setnobreak},
	{\dict{subsix}}{\unskip}{\unskip}{\app{\dict{inrpartial}}{(\app{\dict{symbols}}{\Sigma_1}), (\app{\dict{symbols}}{\Sigma_2})}},
	{\ldots\textit{_no_junk}}{\unskip}{\unskip}{\ldots\setnobreak},
    {\ldots\textit{_agree}}{\unskip}{\unskip}{\ldots}
}

The no junk conditions hold by the no junk property of the sum type, while the agreement conditions are provable vacuously, since no sorts or symbols appear in both constituent domains. We provide a further example instance in \Cref{sec:example}.

If two signatures fulfill the \dict{gluable} property, it is possible to combine them automatically.

\begin{definition}[Signature combining]\label{def:signaturecombining} We construct a composite signature: its sorts and symbols are based on the given instance, and the functions are deferred to the parameter signatures' functions.
\fundef*{\dict{glue}}{\fa {\Sigma_1, \Sigma_2} \app{\dict{gluable}}{\Sigma_1, \Sigma_2} \to \texttt{Signature}}{
    {\dict{Sorts}}{G}{\unskip}{\app{\dict{SharedSort}}{G}},
    {\dict{EV}}{G}{s}{\othercase
        {\app{\dict{EV}}{\Sigma_1, s'}}
        {\te {s'} \app{\dict{partialg}}{(\app{\dict{subone}}{G}), s} = \app{\dict{Some}}{s'}}
        {\app{\dict{EV}}{\Sigma_2, s'}}
        [\te {s'} \app{\dict{partialg}}{(\app{\dict{subtwo}}{G}), s} = \app{\dict{Some}}{s'}]
    },
    {\dict{SV}}{G}{s}{\othercase
        {\app{\dict{SV}}{\Sigma_1, s'}}
        {\te {s'} \app{\dict{partialg}}{(\app{\dict{subone}}{G}), s} = \app{\dict{Some}}{s'}}
        {\app{\dict{SV}}{\Sigma_2, s'}}
        [\te {s'} \app{\dict{partialg}}{(\app{\dict{subtwo}}{G}), s} = \app{\dict{Some}}{s'}]
    },
    {\dict{symbols}}{G}{\unskip}{\app{\dict{SharedSymbols}}{G}},
    {\dict{params}}{G}{\sigma}{\othercase
        {\map{(\app{\dict{partialf}}{(\app{\dict{subone}}{G})})}{(\app{\dict{params}}{\Sigma_1, \sigma'})}}
        {\begin{aligned} &\te {\sigma'} \app{\dict{partialg}}{(\app{\dict{subfive}}{G}), \sigma} = \breakrow\app{\dict{Some}}{\sigma'}\end{aligned}}
        {\map{(\app{\dict{partialf}}{(\app{\dict{subtwo}}{G})})}{(\app{\dict{params}}{\Sigma_2, \sigma'})}}
        [\begin{aligned} &\te {\sigma'} \app{\dict{partialg}}{(\app{\dict{subsix}}{G}), \sigma} = \breakrow\app{\dict{Some}}{\sigma'}\end{aligned}]
    },
    {\dict{return}}{G}{\sigma}{\othercase
        {\app{\dict{partialf}}{(\app{\dict{subone}}{G}), (\app{\dict{return}}{\Sigma_1, \sigma'})}}
        {\te {\sigma'} \app{\dict{partialg}}{(\app{\dict{subfive}}{G}), \sigma} = \app{\dict{Some}}{\sigma'}}
        {\app{\dict{partialf}}{(\app{\dict{subtwo}}{G}), (\app{\dict{return}}{\Sigma_2, \sigma'})}}
        [\te {\sigma'} \app{\dict{partialg}}{(\app{\dict{subsix}}{G}), \sigma} = \app{\dict{Some}}{\sigma'}]
    }
}

Note that due to the no junk and agreement axioms, the functions above are total and deterministic (up to isomorphism).
\end{definition}

This binary operator, defined over signatures, fulfills the aforementioned properties of the sum type. Commutativity holds due to the symmetric nature of the partial isomorphisms allowing them to be interchanged, the commutativity of logical disjunction, and the agreement conditions ensuring that regardless of which implementation the combining function chooses, they are equal.

\begin{\unprovedlemma}[Combining is commutative] Given two combinable signatures, they are combinable the other way around as well, and combining will produce equivalent signatures.
\begin{align*}
    \fa {\Sigma_1, \Sigma_2} &(H_1 : \app{\dict{gluable}}{\Sigma_1, \Sigma_2}) \to\breakrow \DPair{(\app{\dict{gluable}}{\Sigma_2, \Sigma_1})}{H_2}{\app{\dict{glue}}{\Sigma_1, \Sigma_2, H_1} \cong \app{\dict{glue}}{\Sigma_2, \Sigma_1, H_2}}
\end{align*}
\end{\unprovedlemma}

Associativity holds for a similar reason to commutativity; however, to state it, the combined signature must be further combinable.

\begin{\unprovedlemma}[Combining is associative] If two signatures are combinable, and their combination is combinable with another, the order of the operations can be reversed.
\begin{align*}
    \fa {\Sigma_1, \Sigma_2, \Sigma_3} &(H_1 : \app{\dict{gluable}}{\Sigma_1, \Sigma_2}) \to (H_2 : \app{\dict{gluable}}{(\app{\dict{glue}}{\Sigma_1, \Sigma_2, H_1}), \Sigma_3}) \to \breakrow \DPair{(\app{\dict{gluable}}{\Sigma_2, \Sigma_3})}{H_3}{\DPair{(\app{\dict{gluable}}{\Sigma_1, (\app{\dict{glue}}{\Sigma_2, \Sigma_3, H_3})}) \breakrow}{H_4}{\app{\dict{glue}}{(\app{\dict{glue}}{\Sigma_1, \Sigma_2, H_1}), \Sigma_3, H_2} \cong \app{\dict{glue}}{\Sigma_1, (\app{\dict{glue}}{\Sigma_2, \Sigma_3, H_3}), H_4}}}
\end{align*}
\end{\unprovedlemma}

Furthermore, due to the no junk properties, combining a signature with itself produces an equivalent signature. This property is called idempotence.

\begin{\unprovedlemma}[Combining is idempotent] A signature is combinable with itself and the result is equivalent to the input.
\begin{equation*}
    \fa \Sigma \DPair{(\app{\dict{gluable}}{\Sigma, \Sigma})}{H}{\app{\dict{glue}}{\Sigma, \Sigma, H} \cong \Sigma}
\end{equation*}
\end{\unprovedlemma}

The set of signatures, the binary operation over them, and the three properties form a semilattice, where the gluing function is essentially the join of signatures. Such a semilattice induces a partial order~\cite{latticebook}. Intuitively, a signature that is greater than another can be used to create every pattern of every sort that the lesser signature supports, potentially more.

\begin{definition}[Partial order over signatures]\label{def:sigorder} A signature is considered greater than another if they are combinable, and combining them produces a signature equivalent to it.
\begin{equation*}
    \Sigma_1 \preceq \Sigma_2 \coloneqq \DPair{(\app{\dict{gluable}}{\Sigma_1, \Sigma_2})}{H}{\app{\dict{glue}}{\Sigma_1, \Sigma_2, H} \cong \Sigma_2}
\end{equation*}
\end{definition}

The following statement holds about this relation.

\begin{proposition}[The combination is greater than its constituents]\label{prop:combgreater} $\fa {\Sigma_1, \Sigma_2} (H : \app{\dict{gluable}}{\Sigma_1, \Sigma_2}) \to \Sigma_1 \preceq \app{\dict{glue}}{\Sigma_1, \Sigma_2, H} \land \Sigma_2 \preceq \app{\dict{glue}}{\Sigma_1, \Sigma_2, H}$.
\end{proposition}

Note that it is possible to define an ``empty signature'' with no sorts, incapable of hosting any pattern. While this is hardly useful in practice, it can serve as the lower bound for the semilattice and its order.

\subsection{Syntax lifting}\label{sec:patterns}

The syntax of the logic undergoes no structural changes in response to signature combination; however, patterns expressed using different signatures are incompatible with one another. This makes it impossible to check if a pattern is in a set of axioms, for example. The relation induced by the combining operation of the signatures can be used to reconstruct patterns in ``lesser'' signatures into those in ``greater'' ones. We define the \emph{lifting} function as follows.

\begin{definition}[Lifting patterns]\label{def:lift} We define lifting using recursive descent.
\fundef!{\dict{liftone}}{\fa {\Sigma_1, \Sigma, ex, mu, s, (H : \Sigma_1 \preceq \Sigma)} \sigsortedclosedpattern{\Sigma_1}{s}{ex}{mu} \to \breakrow\sigsortedclosedpattern{\Sigma}{(\app{\dict{partialf}}{(\app{\dict{subone}}{\fst{H}}), s})}{(\map{(\app{\dict{partialf}}{(\app{\dict{subone}}{\fst{H}})})}{ex})}{(\map{(\app{\dict{partialf}}{(\app{\dict{subone}}{\fst{H}})})}{mu})}}{
	{(\fevar x)}{\fevar x\setnobreak},
	{(\fsvar X)}{\fsvar X\setnobreak},
	{(\bevar n)}{\bevar{\app{\dict{Inmap}}{n}}\setnobreak},
	{(\bsvar N)}{\bsvar{\app{\dict{Inmap}}{N}}},
	{(\syapp \sigma {xs})}{\syapp{\app{\dict{partialf}}{(\app{\dict{subfive}}{\fst{H}}), \sigma}}{(\hlistmap{\dict{liftone}}{xs}})\setnobreak},
	{(\exists_{s'} \ld \varphi)}{\exists_{\app{\dict{partialf}}{(\app{\dict{subone}}{\fst{H}}), s'}} \ld \app{\dict{liftone}}{\varphi}},
	{(\varphi_1 \land \varphi_2)}{\app{\dict{liftone}}{\varphi_1} \land \app{\dict{liftone}}{\varphi_2}\setnobreak},
	{(\lnot \varphi)}{\lnot (\app{\dict{liftone}}{\varphi})\setnobreak},
	{(\mu \ld \varphi)}{\mu \ld \app{\dict{liftone}}{\varphi}}
}
\end{definition}

The definition above appears simple, with the pattern being simply reconstructed recursively, with only the implicit parameters changing to match the modified sorts. However, there are certain implicit simplifications that enable the input values, such as the $x$ element variable, to be reused in the definition. The explanation of this case can be found in \Cref{app:simpl}. Furthermore, we can also define an inverse to lifting (called \dict{unlift}), which we explain in more detail in \Cref{app:unlift}.

\begin{addtoappendix}
\section{Implicit simplification in \texorpdfstring{\dict{liftone}}{lift}}\label{app:simpl}

As mentioned after \Cref{def:lift}, sometimes some simplification is needed to ensure the type correctness of the definition. In this example, in the context of $(H : \Sigma_1 \preceq \Sigma)$, we explain how the input variable $x$ of type $\app{\dict{EV}}{\Sigma, (\app{\dict{partialf}}{(\app{\dict{subone}}{\fst H}), s})}$ can be used where a variable of type $\app{\dict{EV}}{\Sigma_1, s}$ is needed.
\begin{align*}
    \app{\dict{EV}}{\Sigma, (\app{\dict{partialf}}{(\app{\dict{subone}}{\fst H}), s})} &\eqby{\snd{H}} \\
    \app{\dict{EV}}{(\app{\dict{glue}}{\Sigma_1, \Sigma}), (\app{\dict{partialf}}{(\app{\dict{subone}}{\fst H}), s})} &\eqby{\Cref{def:signaturecombining}} \\
    \othercase
        {\app{\dict{EV}}{\Sigma_1, s'}}
        {\te {s'} \app{\dict{partialg}}{(\app{\dict{subone}}{\fst{H}}), (\app{\dict{partialf}}{(\app{\dict{subone}}{\fst H}), s})} = \app{\dict{Some}}{s'}}
        \ldots[\ldots]  &\eqby{\app{\dict{partialgf}}{\fst{H}}\textit{ from Sec. \ref{operators}}} \\
    \othercase
        {\app{\dict{EV}}{\Sigma_1, s'}}
        {\te {s'} \app{\dict{Some}}{s} = \app{\dict{Some}}{s'}}
        \ldots[\ldots]  &= \\
    \app{\dict{EV}}{\Sigma_1, s} &
\end{align*}
\end{addtoappendix}

\begin{addtoappendix}
\section{The \texorpdfstring{\dict{unliftone}}{unlift} function}\label{app:unlift}
Similarly to pattern lifting shown in \Cref{def:lift}, an \emph{unlifting} function may be defined, which is the optional inverse of lifting.

\begin{definition}[Unlifting patterns] We define the function via recursive descent, and case separation.
\fundef{\dict{unliftone}}{\fa {\Sigma_1, \Sigma, ex, mu, s, (H : \Sigma_1 \preceq \Sigma)} \breakrow\sigsortedclosedpattern{\Sigma}{(\app{\dict{partialf}}{(\app{\dict{subone}}{\fst{H}}), s})}{(\map{(\app{\dict{partialf}}{(\app{\dict{subone}}{\fst{H}})})}{ex})}{(\map{(\app{\dict{partialf}}{(\app{\dict{subone}}{\fst{H}})})}{mu})} \to \breakrow\app{\dict{option}}{(\sigsortedclosedpattern{\Sigma_1}{s}{ex}{mu})}}{
	{(\fevar x)}{\app{\dict{Some}}{\fevar x}\qquad\qquad\app{\dict{unliftone}}{(\bevar n)} \coloneqq \app{\dict{Some}}{(\bevar{\app{\dict{Inmapr}}{\dict{toinj}, n}})}},
	{(\fsvar X)}{\app{\dict{Some}}{\fsvar X}\qquad\qquad\app{\dict{unliftone}}{(\bsvar N)} \coloneqq \app{\dict{Some}}{(\bsvar{\app{\dict{Inmapr}}{\dict{toinj}, N}})}},
	{(\syapp \sigma {\sephetlist{\varphi_1, \ldots, \varphi_n}})}{\othercase
        {\app{\dict{Some}}{(\syapp{\sigma}{\sephetlist{\varphi_1', \ldots, \varphi_n'}})}}
        {\begin{aligned} &\te {\varphi_1', \ldots, \varphi_n'}\bigwedge\limits_{\mathclap{i=1 \ldots n}} \app{\dict{unliftone}}{\varphi_i} = \app{\dict{Some}}{\varphi_i'}\end{aligned}}
        {\dict{None}}
    },
	{(\lnot \varphi)}{\othercase
        {\app{\dict{Some}}{\lnot \varphi'}}
        {\te {\varphi'} \app{\dict{unliftone}}{\varphi} = \app{\dict{Some}}{\varphi'}}
        {\dict{None}}
    },
	{(\varphi_1 \land \varphi_2)}{\othercase
        {\app{\dict{Some}}{\varphi_1' \land \varphi_2'}}
        {\begin{aligned} &\te {\varphi_1', \varphi_2'} \app{\dict{unliftone}}{\varphi_1} = \app{\dict{Some}}{\varphi_1} \land \breakrow\app{\dict{unliftone}}{\varphi_2} = \app{\dict{Some}}{\varphi_2}\end{aligned}}
        {\dict{None}}
    },
	{(\exists_{s'} \ld \varphi)}{\othercase
        {\app{Some}{\exists_{s''} \ld \varphi'}}
        {\begin{aligned} &\te {s'', \varphi'} \app{\dict{partialg}}{(\app{\dict{subone}}{\fst{H}}), s'} = \app{\dict{Some}}{s''} \land \breakrow\app{\dict{unliftone}}{\varphi} = \app{\dict{Some}}{\varphi'}\end{aligned}}
        {\dict{None}}
    },
	{(\mu \ld \varphi)}{\othercase
        {\app{Some}{\mu \ld \varphi'}}
        {\te {\varphi'} \app{\dict{unliftone}}{\varphi} = \app{\dict{Some}}{\varphi'}}
        {\dict{None}}
    }
}
The definition depends on the following lemma.

\end{definition}

\begin{lemma}[\dict{toinj}---injectivity of \dict{partialf}] Given a partial isomorphism \textit{sub}, $\fa {x, y} \app{\dict{partialf}}{\textit{sub}, x} = \app{\dict{partialf}}{\textit{sub}, y} \to x = y$.
\end{lemma}

\begin{proof}
It is proven by applying the congruence of equality with $\app{\dict{partialg}}{sub}$ to the hypothesis, and simplifying both sides using \app{\dict{partialgf}}{sub}, resulting in $\app{\dict{Some}}{x} = \app{\dict{Some}}{y}$. By the injectivity of \dict{Some}, we arrive at the expected conclusion.
\end{proof}

The definition of \dict{unliftone} has three major cases. The first consists of the bases cases, the first four, which always successfully produce a value that is a reconstruction similar to \dict{liftone}. Bound variables---since they are represented by an \dict{In} proof---need to be adjusted using \Cref{lem:inmapr} (\dict{Inmapr}). The second major case encompasses the simple recursive cases (i.e., application, negation, conjunction, and $\mu$ binder), which succeed only if all recursive calls do, and only reconstruct the pattern from the recursive results. Application is of special interest among them, wherein the type of the argument pattern mandates that the return sort of the symbol must necessarily be $(\app{\dict{partialf}}{(\app{\dict{subone}}{\fst{H}}), s})$. As per \Cref{def:signaturecombining}, this is only possible if $\sigma$ is a symbol of $\Sigma_1$, therefore it can be used to reconstruct the pattern in that signature.

The third major case is the existential quantifier. This is the only one that can fail on its own merit, and is the reason why an optional pattern is returned. This also shows that a combined signature can host more patterns than its constituents combined, and explains how: by allowing quantifiers to use bound variables with the greater signature's sort (which is not necessarily included in the smaller signature). 

\begin{proposition}[Partial isomorphism for patterns] Given suitable signatures $\Sigma_1$ and $\Sigma$, the \dict{liftone} and \dict{unliftone} functions define a partial isomorphism over patterns.
\begin{align*}
    &(H : \Sigma_1 \preceq \Sigma) \to \app{\dict{partial}}{(\sigsortedclosedpattern{\Sigma_1}{s}{ex}{mu}), \breakrow(\sigsortedclosedpattern{\Sigma}{(\app{\dict{partialf}}{(\app{\dict{subone}}{\fst{H}}), s})}{(\map{(\app{\dict{partialf}}{(\app{\dict{subone}}{\fst{H}})})}{ex})}{(\map{(\app{\dict{partialf}}{(\app{\dict{subone}}{\fst{H}})})}{mu})})}
\end{align*}
\end{proposition}
\end{addtoappendix}

\subsection{Model gluing}

To combine models, we utilize a similar strategy to \Cref{def:sigcombinecond}. First, to ensure that the gluing can be done in any order, we must ensure that the models agree on the functions given shared inputs. We can define another type class for this purpose.

\begin{definition}[Gluable models] Two models are gluable if their signatures are combinable, and if a sort (symbol resp.) appears in both signatures, their carrier sets (interpretation resp.) are equivalent.

\medskip
\noindent
\typedef[class]{\app{\dict{gluableM}}{\Sigma_1, \Sigma_2, M_1, M_2}}{
    {\dict{combinable}}{\app{\dict{gluable}}{\Sigma_1, \Sigma_2}},
    {\dict{agreefive}}{\fa {s, s_1, s_2} \app{\dict{partialg}}{(\app{\dict{subone}}{\dict{combinable}}), s} = \app{\dict{Some}}{s_1} \to \breakrow\app{\dict{partialg}}{(\app{\dict{subtwo}}{\dict{combinable}}), s} = \app{\dict{Some}}{s_2} \to \breakrow\app{\dict{carrier}}{M_1, s_1} \cong \app{\dict{carrier}}{M_2, s_2}},
    {\dict{agreesix}}{\fa {\sigma, \sigma_1, \sigma_2, xs} \app{\dict{partialg}}{(\app{\dict{subfive}}{\dict{combinable}}), \sigma} = \app{\dict{Some}}{\sigma_1} \to \breakrow\app{\dict{partialg}}{(\app{\dict{subsix}}{\dict{combinable}}), \sigma} = \app{\dict{Some}}{\sigma_2} \to \breakrow\app{\dict{app}}{M_1, \sigma_1, xs} \cong \app{\dict{app}}{M_2, \sigma_2, xs}}
}

We omit agreement on \dict{inh}, since, as mentioned at the start of this section, its specific value is irrelevant and not considered for equivalence.
\end{definition}

Based on this type class, we define a gluing function for models in the following way. For this definition, we can define the following conversion by induction:
\begin{equation}\label{lem:hlistmapmove}\tag{\dict{hlistmapmove}}
    \app{\dict{HList}}{A, F, (\map{f}{l})} \leftrightarrow \app{\dict{HList}}{A, (F \circ f), l}
\end{equation}

\begin{definition}[Model gluing]\label{def:glueM} We define the glued model by dispatching its operations to its constituents.
\fundef*{\dict{glueM}}{\fa {\Sigma_1, \Sigma_2, M_1, M_2, (H : \app{\dict{gluableM}}{\Sigma_1, \Sigma_2, M_1, M_2})} \app{\dict{Model}}{(\app{\dict{glue}}{\Sigma_1, \Sigma_2})}}{
    {\dict{carrier}}{\unskip}{s}{\othercase
        {\app{\dict{carrier}}{M_1, s'}}
        {\te {s'} \app{\dict{partialg}}{(\app{\dict{subone}}{(\app{\dict{combinable}}{H})}), s} = \app{\dict{Some}}{s'}}
        {\app{\dict{carrier}}{M_2, s'}}
        [\te {s'} \app{\dict{partialg}}{(\app{\dict{subtwo}}{(\app{\dict{combinable}}{H})}), s} = \app{\dict{Some}}{s'}]
    },
    {\dict{inh}}{\unskip}{s}{\othercase
        {\app{\dict{inh}}{M_1, s'}}
        {\te {s'} \app{\dict{partialg}}{(\app{\dict{subone}}{(\app{\dict{combinable}}{H})}), s} = \app{\dict{Some}}{s'}}
        {\app{\dict{inh}}{M_2, s'}}
        [\te {s'} \app{\dict{partialg}}{(\app{\dict{subtwo}}{(\app{\dict{combinable}}{H})}), s} = \app{\dict{Some}}{s'}]
    },
    {\dict{app}}{\unskip}{\sigma\ \textit{xs}}{\othercase
        {\app{\dict{app}}{M_1, \sigma', (\app{\ref{lem:hlistmapmove}}{\textit{xs}})}}
        {\begin{aligned} &\te {\sigma'} \app{\dict{partialg}}{(\app{\dict{subfive}}{(\app{\dict{combinable}}{H})}), \sigma} = \breakrow\app{\dict{Some}}{\sigma'}\end{aligned}}
        {\app{\dict{app}}{M_2, \sigma', (\app{\ref{lem:hlistmapmove}}{\textit{xs}})}}
        [\begin{aligned} &\te {\sigma'} \app{\dict{partialg}}{(\app{\dict{subsix}}{(\app{\dict{combinable}}{H})}), \sigma} = \breakrow\app{\dict{Some}}{\sigma'}\end{aligned}]
    }
}
The conditions are again total due to the no junk properties.
\end{definition}

Since the combined signature ensures that each sort can be mapped to at least one of the input signatures' sorts, we can defer to the carriers provided by the input models for these sorts. By the type class, we also require them to be isomorphic, if both models have the given sort.

Interpretations of symbol applications may similarly be traced back to the constituent models' methods. By $\dict{subtotalthree}$, the input symbol must be present in at least one of the signatures. Furthermore, by the definition of the combining function (\Cref{def:signaturecombining}), we know that the parameter and return sorts of this symbol are from the same signature, lifted into the combined one. This enables us to reduce the carrier elements (in the argument list \textit{xs}) to the constituent model's elements. With both the symbol and the inputs reduced, it is possible to use the constituent model's interpretation function to produce a set of values in this smaller model's carrier. Now we may utilize the partial isomorphisms again to raise these values back to the glued model's carrier elements. Due to the definition of \dict{return}, this is exactly the type expected to be the result of the glued model's interpretation function.

The above definition is once again simplified for readability. \Cref{app:simpltwo} shows an example of how the $\textit{xs}$ parameter of \dict{app}, as well as the deferred use of $\app{\dict{app}}{M_1}$ are transformed to be well-typed.

\begin{addtoappendix}
\section{Implicit simplification in \texorpdfstring{\dict{app}}{app}}\label{app:simpltwo}

As mentioned after \Cref{def:glueM}, in order for the definition of \dict{app} to work, for example, considerable simplification happens in the background. In the first case, we have $(H : \app{\dict{gluableM}}{\Sigma_1, \Sigma_2, M_1, M_2})$ and the condition $\te {\sigma'} \app{\dict{partialg}}{(\app{\dict{subfive}}{(\app{\dict{combinable}}{H})}), \sigma} = \dict{Some}\ \sigma'$, we can first transform this using the property (of \dict{partial} defined at the beginning of \Cref{operators}) $\app{\dict{partialfg}}{(\app{\dict{subfive}}{(\app{\dict{combinable}}{H})})}$ into $\te {(\sigma' : \app{\dict{symbols}}{\Sigma_1})} \sigma = \app{\dict{partialf}}{(\app{\dict{subfive}}{(\app{\dict{combinable}}{H})}), \sigma'}$. We refer to the symbol of this existential as $\sigma_1$ and the equality as $H_1$. Beyond these, we make use of the following equalities:

\begin{itemize}
    \item the $\eta$ rule of functions: $\fa {A, B, (f : A \to B)} f = \la x \app fx$, provided $x$ does not occur in $f$, where the introduction of the parameter is referred to as \emph{$\eta$-expansion} and its elimination as \emph{$\eta$-reduction}~\cite{LCbook};
    \item the definition of the combined \dict{params} and \dict{return}, from \Cref{def:signaturecombining};
    \item the definition of the glued \dict{carrier}, from \Cref{def:glueM}
\end{itemize}

With these, the type of $xs$ can undergo the following conversions:
\begin{align*}
	\app{\dict{HList}}{\dict{Sorts}, \dict{carrier}, (\app{\dict{params}}{(\app{\dict{glue}}{\Sigma_1, \Sigma_2}), \sigma})} &\eqby{H_1} \\
	\app{\dict{HList}}{\dict{Sorts}, \dict{carrier}, (\app{\dict{params}}{(\app{\dict{glue}}{\Sigma_1, \Sigma_2}), (\app{\dict{partialf}}{(\app{\dict{subfive}}{(\app{\dict{combinable}}{H})}), \sigma_1})})} &\eqby{\dict{params} \text{ def.}} \\
	\app{\dict{HList}}{\dict{Sorts}, \dict{carrier}, (\map{(\app{\dict{partialf}}{(\app{\dict{subone}}{(\app{\dict{combinable}}{H})})})}{(\app{\dict{params}}{\Sigma_1, \sigma_1})})} &
\end{align*}

Contrary to \Cref{app:simpl}, we simplify definitions with case separation without the intermediary step containing the braces and the irrelevant other case. At this stage, \ref{lem:hlistmapmove} can be applied, resulting in the type $\app{\dict{HList}}{\dict{Sorts}, (\la s \app{\dict{carrier}}{(\app{\dict{partialf}}{(\app{\dict{subone}}{(\app{\dict{combinable}}{H})}), s})}), (\app{\dict{params}}{\Sigma_1, \sigma_1})}$, which can be further simplified:
\begin{align*}
	\app{\dict{HList}}{\dict{Sorts}, (\dict{carrier} \circ (\app{\dict{partialf}}{(\app{\dict{subone}}{(\app{\dict{combinable}}{H})})})), (\app{\dict{params}}{\Sigma_1, \sigma_1})} &\eqby{\eta\text{-expansion}} \\
	\app{\dict{HList}}{\dict{Sorts}, (\la s \app{\dict{carrier}}{(\app{\dict{partialf}}{(\app{\dict{subone}}{(\app{\dict{combinable}}{H})}), s})}), (\app{\dict{params}}{\Sigma_1, \sigma_1})} &\eqby{\dict{carrier}\text{ def.}} \\
	\app{\dict{HList}}{\dict{Sorts}, (\la s \app{\dict{carrier}}{s}), (\app{\dict{params}}{\Sigma_1, \sigma_1})} &\eqby{\eta\text{-reduction}} \\
	\app{\dict{HList}}{\dict{Sorts}, \dict{carrier}, (\app{\dict{params}}{\Sigma_1, \sigma_1})} &
\end{align*}

With $xs$ reduced to this type, we can construct $\app{\dict{app}}{\sigma_1, xs}$. This has type $\mathcal P(\app{\dict{carrier}}{(\app{\dict{return}}{\Sigma_1, \sigma_1})}$. This is likewise transformable to the desired return type of the interpretation function:

\begin{align*}
	\mathcal P(\app{\dict{carrier}}{(\app{\dict{return}}{\Sigma_1, \sigma_1})}) &\eqby{\dict{carrier}\text{ def.}} \\
	\mathcal P(\app{\dict{carrier}}{(\app{\dict{partialf}}{(\app{\dict{subone}}{(\app{\dict{combinable}}{H})}), (\app{\dict{return}}{\Sigma_1, \sigma_1})})}) &\eqby{\dict{return}\text{ def.}} \\
	\mathcal P(\app{\dict{carrier}}{(\app{\dict{return}}{(\app{\dict{glue}}{\Sigma_1, \Sigma_2}), (\app{\dict{partialf}}{(\app{\dict{subone}}{(\app{\dict{combinable}}{H})})\ \sigma_1})})}) &\eqby{H_1} \\
	\mathcal P(\app{\dict{carrier}}{(\app{\dict{return}}{(\app{\dict{glue}}{\Sigma_1, \Sigma_2}), \sigma})})
\end{align*}

We note that many type theory systems only allow simplification with definitions, not other equalities, such as $H_1$. Instead, they rely on a concept called \emph{transport}~\cite{hottbook}. Just as these systems attempt to mask this fact from the user, we omit the definition and usage of this function here. Some examples of its usage in this formalization can be found in~\cite{kurucz2025on}.
\end{addtoappendix}

\subsection{Satisfaction}

The first definition necessary for reasoning about combined satisfaction is the combination of theories, definable via the lifting function of patterns (\Cref{sec:patterns}). Given two theories $\Gamma_1$ and $\Gamma_2$ with axioms stated over signatures $\Sigma_1$ and $\Sigma_2$ respectively, we cannot take their union directly because the types of their patterns are indexed by different signatures, and therefore their type is different. However, combining the signatures results in one greater than both (\Cref{prop:combgreater}), which we can utilize as an intermediary in the result type, and apply either the lifting or unlifting to every pattern of the theories to match the unified signature.

\begin{definition}[Merged theories] The merged theory $\app{\dict{glueT}}{\Gamma_1, \Gamma_2}$ can be defined in the following way, using set notation.
\begin{equation*}
    \app{\dict{glueT}}{\Gamma_1, \Gamma_2} \coloneqq \{\app{\dict{liftone}}{\varphi} \mid \varphi \in \Gamma_1\} \cup \{\app{\dict{lifttwo}}{\varphi} \mid \varphi \in \Gamma_2\} \\
\end{equation*}
\end{definition}

\begin{remark*} The following is also a valid definition, based on \dict{unlift} defined in \Cref{app:unlift}.
\begin{equation*}
    \app{\dict{glueT}}{\Gamma_1, \Gamma_2} \coloneqq \{\varphi \mid (\te {\varphi'} \app{\dict{unliftone}}{\varphi} = \app{\dict{Some}}{\varphi'} \land \varphi' \in \Gamma_1) \lor \te {\varphi'} \app{\dict{unlifttwo}}{\varphi} = \app{\dict{Some}}{\varphi'} \land \varphi' \in \Gamma_2 \}
\end{equation*}
\end{remark*}

Both definitions are equivalent, but the former is simpler to use. Due to the sorted nature of the logic and the conditions on combining, the patterns express the same axioms, so this theory is the proper union of the inputs. However, with the axioms residing in the combined signature, it is now theoretically possible for the composite theory to contain axioms that neither of the constituents do. This is helpful for model extension; however, it must be considered when proving satisfaction (\Cref{prop:satisfaction}), by inspecting the combination, and utilizing the fact that in only uses lifting, and does not add anything else.

Before we can evaluate patterns in the glued models, we must address the combination of variable valuations as well. This type behaves differently compared to the others discussed in this section. Defining an unlifting function over them is possible and results in a total definition.

\begin{definition}[Unlifting valuations] By simply converting the sort, the same simplification shown in \Cref{app:simpl,app:simpltwo} applies, and we can dispatch to the greater valuation.
\fundef*{\dict{deval}}{\fa {\Sigma_1, \Sigma, M_1, M, (H : \Sigma_1 \preceq \Sigma)} \app{\dict{Valuation}}{\Sigma, M} \to \app{\dict{Valuation}}{\Sigma_1, M_1}}{
    {\dict{evar_val}}{\rho}{s}{\app{\dict{evar_val}}{\rho, (\app{\dict{partialf}}{(\app{\dict{subone}}{\fst{H}}), s})}},
    {\dict{svar_val}}{\rho}{s}{\app{\dict{svar_val}}{\rho, (\app{\dict{partialf}}{(\app{\dict{subone}}{\fst{H}}), s})}}
}
\end{definition}

Lifting valuations is impossible. We can define two functions that examine the sort and decide whether they can proceed, with types $\fa s \app{\dict{option}}{(\app{\dict{EV}}{s} \to \app{\dict{carrier}}{M, s})}$ and $\fa s \app{\dict{option}}{(\app{\dict{SV}}{s} \to \mathcal P(\app{\dict{carrier}}{M, s}))}$ respectively. These are defined analogously to previous functions, with case separation on \dict{partialg}. However, the sort parameters cannot be moved inside the \dict{option}s, as that would mean the decision has to be made without examining the sort, in which case the only reasonable definitions would be to always return \dict{None} in both, since the results cannot always be determined. These functions are already rather useless; however, if we did define them, the \dict{option}s could move outside the record, resulting in a function of type $\fa {\Sigma_1, \Sigma, M_1, M, (H : \Sigma_1 \preceq \Sigma)} \app{\dict{Valuation}}{\Sigma_1, M_1} \to \app{\dict{option}}{(\app{\dict{Valuation}}{\Sigma, M})}$ that always returns \dict{None}. This, with the first five parameters provided, is now well-typed to serve as the \dict{partialg} of $\app{\dict{partial}}{(\app{\dict{Valuation}}{\Sigma, M}), (\app{\dict{Valuation}}{\Sigma_1, M_1})}$, where \dict{partialf} is \dict{deval} with the same parameters; however, this isomorphism's \dict{partialgf} is unsatisfiable, and is therefore not actually a partial isomorphism.

Finally, we present the main property of gluing, and provide the proof sketch for it in \Cref{app:satisfactionsketch}.

\begin{proposition}[Gluing preserves satisfaction]\label{prop:satisfaction}
    Given two models $M_1$ and $M_2$, each satisfying a theory $\Gamma_1$ and $\Gamma_2$, if they are gluable, the glued model constructed from these two satisfies the glued theory: $M_1 \models \Gamma_1 \to M_2 \models \Gamma_2 \to \app{\dict{glueM}}{M_1, M_2} \models \app{\dict{glueT}}{\Gamma_1, \Gamma_2}$.
\end{proposition}

\begin{addtoappendix}
\section{Proof sketch of satisfaction preservation}\label{app:satisfactionsketch}

We provide the following proof sketch for \Cref{prop:satisfaction}.

\begin{proof}[Proof sketch]
    The two models, $M_1$ and $M_2$, satisfying their theories, $\Gamma_1$ and $\Gamma_2$ means that (with all parameters of \dict{eval} explicitly provided for clarity),
    \begin{align*}
        \fa {s_1, \varphi_1, \rho_1} \app{\dict{existT}}{s_1, \varphi_1} \in \Gamma_1 \to \app{\dict{eval}}{\Sigma_1, [], [], s_1, M_1, \rho_1, \varphi_1} = \topset{\app{\dict{carrier}}{M_1, s_1}} &\text{ and} \\
        \fa {s_2, \varphi_2, \rho_2} \app{\dict{existT}}{s_2, \varphi_2} \in \Gamma_2 \to \app{\dict{eval}}{\Sigma_2, [], [], s_2, M_2, \rho_2, \varphi_2} = \topset{\app{\dict{carrier}}{M_2, s_2}} &\text{.}
    \end{align*}
    
    Similarly, proving that the glued model satisfies the glued theory means:
    \begin{equation*}
        \fa {s, \varphi, \rho} \app{\dict{existT}}{s, \varphi} \in \app{\dict{glueT}}{\Gamma_1, \Gamma_2} \to \app{\dict{eval}}{\Sigma, [], [], s, (\app{\dict{glueM}}{M_1, M_2}), \rho, \varphi} = \topset{\app{\dict{carrier}}{(\app{\dict{glueM}}{M_1, M_2}), s}}
    \end{equation*}

    Here, $\app{\dict{existT}}{s, \varphi} \in \app{\dict{glueT}}{\Gamma_1, \Gamma_2}$ means that either
    \begin{align*}
        \te {s_1, \varphi_1} \app{\dict{existT}}{s_1, \varphi_1} \in \Gamma_1 \land \app{\dict{partialf}}{\dict{subone}, s_1} = s \land \app{\dict{liftone}}{\varphi_1} = \varphi &\text{ or} \\
        \te {s_2, \varphi_2} \app{\dict{existT}}{s_2, \varphi_2} \in \Gamma_2 \land \app{\dict{partialf}}{\dict{subtwo}, s_2} = s \land \app{\dict{lifttwo}}{\varphi_2} = \varphi &\text{.}
    \end{align*}

    Using these and \dict{deval} to instantiate the previous statements we get
    \begin{align*}
        \app{\dict{eval}}{\Sigma_1, [], [], s_1, M_1, (\app{\dict{deval}}{\rho}), \varphi_1} = \topset{\app{\dict{carrier}}{M_1, s_1}} &\text{ or} \\
        \app{\dict{eval}}{\Sigma_2, [], [], s_2, M_2, (\app{\dict{deval}}{\rho}), \varphi_2} = \topset{\app{\dict{carrier}}{M_2, s_2}} &\text{.}
    \end{align*}

    In both cases, the equalities over $\varphi$ can be used to perform further rewriting in the statement to be proved, eliminating some variables. Then the generalized version of this proof, with arbitrary $ex$ and $mu$ environments, can be completed via induction on the remaining pattern, allowing all previous definitions that appear here to be simplified.
\end{proof}
\end{addtoappendix}

\subsection{Example}\label{sec:example}

We present an example usage of these type classes. Let us define two similar signatures $\Sigma_1$ and $\Sigma_2$:

\smallskip\hfill%
\begin{minipage}{.48\linewidth}
\noindent
\fundef*{\Sigma_1}{\texttt{Signature}}{
    {\dict{Sorts}}{\unskip}{\unskip}{\enumcons{\dict{nat}, \dict{bool}}},
    {\dict{EV}}{\unskip}{s}{\dict{string}},
    {\dict{SV}}{\unskip}{s}{\dict{string}},
    {\dict{symbols}}{\unskip}{\unskip}{\enumcons{\dict{isEven}, \dict{twice}, \dict{plus}}},
    {\dict{params}}{\unskip}{\dict{isEven}}{\seplist {\dict{nat}} *},
    {\dict{params}}{\unskip}{\dict{twice}}{\seplist {\dict{nat}} *},
    {\dict{params}}{\unskip}{\dict{plus}}{\seplist {\dict{nat}} {\dict{nat}} *},
    {\dict{return}}{\unskip}{\dict{isEven}}{\dict{bool}},
    {\dict{return}}{\unskip}{\dict{twice}}{\dict{nat}},
    {\dict{return}}{\unskip}{\dict{plus}}{\dict{nat}}
}
\end{minipage}%
\begin{minipage}{.48\linewidth}
\noindent
\fundef*{\Sigma_2}{\texttt{Signature}}{
    {\dict{Sorts}}{\unskip}{\unskip}{\enumcons{\dict{pos}, \dict{bool}}},
    {\dict{EV}}{\unskip}{s}{\dict{string}},
    {\dict{SV}}{\unskip}{s}{\dict{string}},
    {\dict{symbols}}{\unskip}{\unskip}{\enumcons{\dict{isEven}, \dict{thrice}, \dict{add}}},
    {\dict{params}}{\unskip}{\dict{isEven}}{\seplist {\dict{pos}} *},
    {\dict{params}}{\unskip}{\dict{thrice}}{\seplist {\dict{pos}} *},
    {\dict{params}}{\unskip}{\dict{add}}{\seplist {\dict{pos}} {\dict{pos}} *},
    {\dict{return}}{\unskip}{\dict{isEven}}{\dict{bool}},
    {\dict{return}}{\unskip}{\dict{thrice}}{\dict{pos}},
    {\dict{return}}{\unskip}{\dict{add}}{\dict{pos}}
}
\end{minipage}%
\hfill\medskip

We can provide models $M_1$ and $M_2$ for them as well:

\smallskip\hfill%
\begin{minipage}{.48\linewidth}
\noindent
\fundef*{M_1}{\app{\texttt{Model}}{\Sigma_1}}{
    {\dict{carrier}}{\unskip}{\dict{nat}}{\dict{metanat}},
    {\dict{carrier}}{\unskip}{\dict{bool}}{\dict{metabool}},
    {\dict{app}}{\unskip}{\dict{isEven}\ \sephetlist{n}}{\othercase{\{\dict{metatrue}\}}{n\text{ is even}}{\{\dict{metafalse}\}}},
    {\dict{app}}{\unskip}{\dict{twice}\ \sephetlist{n}}{\{2 \cdot n\}},
    {\dict{app}}{\unskip}{\dict{plus}\ \sephetlist{n, m}}{\{n + m\}}
}
\end{minipage}%
\begin{minipage}{.48\linewidth}
\noindent
\fundef*{M_2}{\app{\texttt{Model}}{\Sigma_2}}{
    {\dict{carrier}}{\unskip}{\dict{pos}}{\dict{metapos}},
    {\dict{carrier}}{\unskip}{\dict{bool}}{\dict{metanat}},
    {\dict{app}}{\unskip}{\dict{isEven}\ \sephetlist{n}}{\othercase{\{1\}}{n\text{ is even}}{\{0\}}},
    {\dict{app}}{\unskip}{\dict{thrice}\ \sephetlist{n}}{\{3 \cdot n\}},
    {\dict{app}}{\unskip}{\dict{add}\ \sephetlist{n, m}}{\{n + m\}}
}
\end{minipage}%
\hfill\medskip

The first model is a standard interpretation of the signature, while the second uses a C-style interpretation for booleans. This incompatibility cannot be resolved during gluing, but the shared number type and the addition function can be merged, allowing for an instantiation better than the union-based one.

We use the abstract notation $\enumcons{\textit{con}_1, \ldots, \textit{con}_n}$ to denote an enumeration-like type, which can be defined inductively, or simulated with an isomorphic type. Partial isomorphisms are given as a set of mappings $v_i^f \mapsto v_i^t$, where $v_i^f$ ranges over the entire domain of the lesser type. For every input $v_i^f$, \dict{partialf} returns $v_i^t$, while \dict{partialg} does the opposite, except if its input has no corresponding mapping, \dict{None} is returned. The conditions of \dict{partial} (see the start of \Cref{operators}) for such a definition hold trivially.
\fundef*{C_\Sigma}{\app{\dict{gluable}}{\Sigma_1, \Sigma_2}}{
    {\dict{SharedSort}}{\unskip}{\unskip}{\enumcons{\dict{gnum}, \dict{gbool}, \dict{gcbool}}},
    {\dict{subone}}{\unskip}{\unskip}{\mapthese{{\dict{nat}}{\dict{gnum}},{\dict{bool}}{\dict{gbool}}}},
    {\dict{subtwo}}{\unskip}{\unskip}{\mapthese{{\dict{pos}}{\dict{gnum}},{\dict{bool}}{\dict{gcbool}}}},
    {\dict{SharedSymbols}}{\unskip}{\unskip}{\enumcons{\dict{gisEven}, \dict{gcisEven}, \dict{gtwice}, \dict{gthrice}, \dict{gplus}}},
    {\dict{subfive}}{\unskip}{\unskip}{\mapthese{{\dict{isEven}}{\dict{gisEven}}, {\dict{twice}}{\dict{gtwice}}, {\dict{plus}}{\dict{gplus}}}},
    {\dict{subsix}}{\unskip}{\unskip}{\mapthese{{\dict{isEven}}{\dict{gcisEven}}, {\dict{thrice}}{\dict{gthrice}}, {\dict{add}}{\dict{gplus}}}},
    \ldots\unskip\unskip\ldots
}

The remaining conditions are easily proven: every sort and symbol is mapped to, therefore both $\dict{subtotalone}$ and $\dict{subtotalthree}$ hold; every variable type is defined as \dict{string}, so $\dict{agreeone}$ and $\dict{agreetwo}$ also hold, regardless of their parameters; and the only shared symbol, $\dict{gplus}$, has the same parameter and return sorts with respect to the isomorphisms as its ancestors in the constituent signatures, proving $\dict{agreethree}$ and $\dict{agreefour}$.

Note that at this stage, it would have been possible to map the \dict{bool}s to the same sort, and with them, to unite the remaining two symbols as well. The combined signature resulting from such an instance can be used for syntactical purposes, and, given models with appropriate interpretations, even for semantic reasoning. However, the instance defined this way cannot be used to glue the specific models presented here, since the models' functions do not agree: the \dict{bool}s have non-isomorphic carrier sets; with them not being identifiable, the \dict{isEven}s cannot be either (due to differing return sorts); and \dict{twice} and \dict{thrice} have a different interpretation, meaning they also must remain separated.

The remaining unifiable sorts, \dict{nat} and \dict{pos} have isomorphic carrier sorts, as each natural can be uniquely mapped to its successor, and each positive to its predecessor. The interpretations of \dict{plus} and \dict{add} agree, as can be seen on the definition of the models; therefore, despite their differing names, they can be treated as the same symbol. $C_\Sigma$ and the arguments outlined in this paragraph can be combined to define $G_M : \app{\dict{gluableM}}{\Sigma_1, \Sigma_2, M_1, M_2}$.

Now we may define two theories, $\Gamma_1$ and $\Gamma_2$, satisfied by the models $M_1$ and $M_2$ respectively. We use the following standard notations: $\varphi_1 \lor \varphi_2 \coloneqq \lnot (\lnot \varphi_1 \land \lnot \varphi_2)$ and $\varphi_1 \to \varphi_2 \coloneqq \lnot \varphi_1 \lor \varphi_2$.
\begin{align*}
	\Gamma_1 &= \{\exists \ld \mu \ld \bevar0 \lor \syapp {\dict{plus}} {\sephetlist{\bevar0, \bsvar0}},\ \exists \ld \syapp {\dict{isEven}} {\sephetlist{\bevar0}}\} \\
	\Gamma_2 &= \{\exists \ld \mu \ld \bevar0 \lor \syapp {\dict{add}} {\sephetlist{\bevar0, \bsvar0}},\ \syapp {\dict{isEven}} {\sephetlist{\fevar x}} \to \syapp {\dict{isEven}} {\sephetlist{\syapp {\dict{thrice}} {\sephetlist{\fevar x}}}}\}
\end{align*}

These axioms evaluate to the following sets:
\begin{equation*}
    \begin{array}{@{} r @{\mskip 24mu\relax} c @{\mskip 32mu\relax} c @{}}
        \Gamma_1: & \bigcup\limits_{x \in \dict{metanat}} \app{\lfp}{(\la Y \{x\} \cup \{x + y \mid y \in Y\})} & \smash[t]{\bigcup\limits_{x \in \dict{metanat}} \othercase{\{\dict{metatrue}\}}{x\text{ is even}}{\{\dict{metafalse}\}}} \\
        \Gamma_2: & \bigcup\limits_{x \in \dict{metapos}} \app{\lfp}{(\la Y \{x\} \cup \{x + y \mid y \in Y\})} & \left(\tinynegspace\dict{metanat} \setminus \othercase{\{1\}}{n\text{ is even}}{\{0\}}\right) \cup \left(\tinynegspace\othercase{\{1\}}{3 \cdot n\text{ is even}}{\{0\}}\right) \\
        & & \hphantom{pushrightalittle}\text{where } \{n\} = \app{\dict{evar_val}}{\rho, \dict{bool}, x}
    \end{array}
\end{equation*}

The first axiom has sort \dict{nat} in the first theory and \dict{pos} in the second, as dictated by the two signatures, but both evaluate to a similar set. The fixpoint unfolds to $\{x,\ x + x,\ x + x + x,\ \ldots\}$, also known as the positive multiples of $x$. Taking the union of every positive integer's positive multiples is equal to the set of all positive integers trivially, since the positive multiples of 1 already form this set; therefore, $M_2$ satisfies this axiom of $\Gamma_2$. Doing the same over natural numbers also holds, with an additional reasoning step: the positive multiples of 0 is the set $\{0\}$, which can be added to the set of positive integers to form the set of natural numbers. As a result $M_1$ also satisfies this axiom in $\Gamma_1$.

The second axiom of $\Gamma_1$ has sort \dict{bool} and its corresponding carrier elements cover the full set of booleans after consdiering just the first two values. $M_1$ satisfies both this axiom, and all of $\Gamma_1$. Similarly, the second axiom of $\Gamma_2$ is of sort \dict{bool}; however, additional care must be taken, since $M_2$ assigns natural numbers as its carrier. As a result, the previous axiom would not be satisfiable here, as it would evaluate only to $\{0, 1\}$. The axiom we posed instead results in a set, which can be simplified to $\dict{metanat}$---since $n$ is even if and only if $3 \cdot n$ is even (regardless of the valuation $\rho$ used in the definition of $n$)---making both this axiom, and all of $\Gamma_2$ satisfied by $M_2$.

With the two theories established, we can combine them:
\begin{align*}
	\app{\dict{glueT}}{\Gamma_1, \Gamma_2} &= \{\exists \ld \mu \ld \bevar0 \lor \syapp {\dict{gplus}} {\sephetlist{\bevar0, \bsvar0}},\ \exists \ld \syapp {\dict{gisEven}} {\sephetlist{\bevar0}},\ \syapp {\dict{gcisEven}} {\sephetlist{\fevar x}} \to \syapp {\dict{gcisEven}} {\sephetlist{\syapp {\dict{gthrice}} {\sephetlist{\fevar x}}}}\}
\end{align*}

We make several observations:

\begin{itemize}
	\item This theory looks similar to the union of its constituents; however, any renaming that was done during the combining of the signatures applies (e.g., \dict{isEven} from $\Gamma_2$ is now \dict{gcisEven}).
	\item After simplification, lifted theories may contain the same pattern, allowing us eliminate duplications, as is the case with the first axiom of each constituent.
	\item The theory does not contain any new patterns, even though the new type would enable us to write, for example, $\exists \ld \syapp{\dict{gisEven}}{\sephetlist{\syapp{\dict{gtwice}}{\bevar0}}} \land \syapp{\dict{gisEven}}{\sephetlist{\syapp{\dict{gthrice}}{\bevar0}}}$. However, the satisfiability of such axioms may still be examined separately.
\end{itemize}

The satisfiability of the combined theory in the glued model is proven, without the need for additional reasoning, by applying \Cref{prop:satisfaction}. Below is a diagram demonstrating the relationship of the discussed concepts of glued satisfaction on these axioms. On the horizontal axis is the relationship of typical matching logic constructs, while the vertical one shows how these are combined with the functions discussed in this paper for each step.

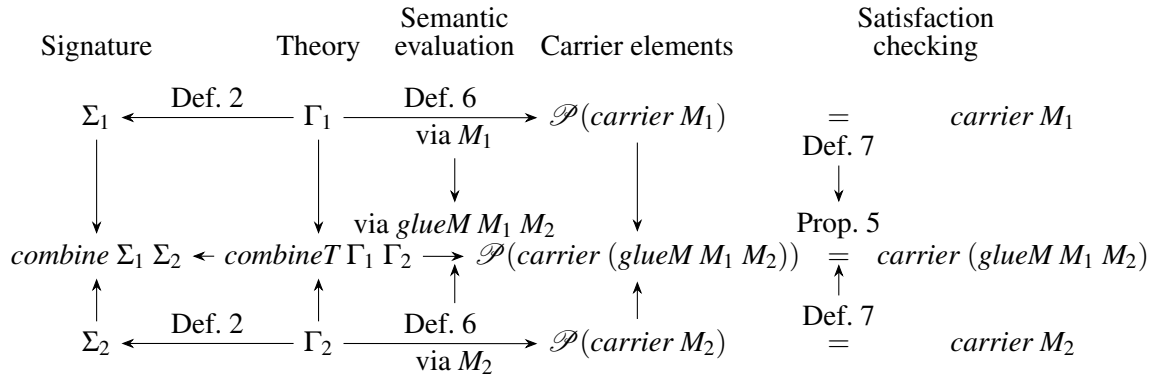
\begingroup
\centering
\begin{tikzpicture}
    \matrix[matrix of math nodes, nodes in empty cells] (m) {
        \text{Signature} &[3mm] \text{Theory} &[3mm] \mathclap{\shortstack{Semantic\\evaluation}} & \text{Carrier elements} &[1mm] &[1mm] \\[3mm]
        \Sigma_1 & \Gamma_1 & |[below=-1.5mm]| \mathclap{\text{via } M_1} & \mathcal P(\app{\dict{carrier}}{M_1}) & \node {=}; \node[below] (m-2-5) {\mathclap{\text{Def.~\ref{def:sat}}}}; & \app{\dict{carrier}}{M_1} \\[5mm]
        \app{\dict{glue}}{\Sigma_1, \Sigma_2} & \app{\dict{glueT}}{\Gamma_1, \Gamma_2} & |[above=2mm]| \mathclap{\text{via } \app{\dict{glueM}}{M_1, M_2}} & \mathcal P(\app{\dict{carrier}}{(\app{\dict{glueM}}{M_1, M_2})}) & \node{=}; \node[above=2mm] (m-3-5) {\mathclap{\text{Prop.~\ref{prop:satisfaction}}}}; & \app{\dict{carrier}}{(\app{\dict{glueM}}{M_1, M_2})} \\[2mm]
        \Sigma_2 & \Gamma_2 & |[below=-1.5mm]| \mathclap{\text{via } M_2} & \mathcal P(\app{\dict{carrier}}{M_2}) & \node {=}; \node[above=2mm] (m-4-5) {\mathclap{\text{Def.~\ref{def:sat}}}}; & \app{\dict{carrier}}{M_2} \\
    };

    \node[anchor=base] at ($(m-1-5)!.5!(m-1-6)$) {\shortstack{Satisfaction\\checking}};

    \draw[-Stealth] (m-2-2) to[edge label'=Def.~\ref{def:pattern}] (m-2-1);
    \draw[-Stealth] (m-2-2) to[edge label=Def.~\ref{def:sem}] (m-2-4);
    \draw[-Stealth] (m-3-2) to (m-3-1);
    \draw[-Stealth] (m-3-2) to (m-3-4);
    \draw[-Stealth] (m-4-2) to[edge label'=Def.~\ref{def:pattern}] (m-4-1);
    \draw[-Stealth] (m-4-2) to[edge node={node[above] (asd) {Def.~\ref{def:sem}}}] (m-4-4);

    \draw[-Stealth] (m-2-1) to (m-3-1);
    \draw[-Stealth] (m-4-1) to (m-3-1);
    \draw[-Stealth] (m-2-2) to (m-3-2);
    \draw[-Stealth] (m-4-2) to (m-3-2);
    \draw[-Stealth] (m-2-3) to (m-3-3);
    \draw[-Stealth] (m-4-3 |- asd.north) to ([yshift=-1mm] m-3-3 |- m-3-2);
    \draw[-Stealth] (m-2-4) to (m-3-4);
    \draw[-Stealth] (m-4-4) to (m-3-4);
    \draw[-Stealth] (m-2-5) to (m-3-5);
    \draw[-Stealth] (m-4-5) to ([yshift=-1mm] m-3-5 |- m-3-4);
\end{tikzpicture}\par
\endgroup

A larger version of this diagram, showing the values for this example can be found in \Cref{app:bigboy}.

\begin{addtoappendix}
\section{Detailed example diagram}\label{app:bigboy}

The following diagram is a more detailed illustration of the relationship of matching logic's concepts and the functions defined over them in this paper. The structures used here match the ones defined in \Cref{sec:example}.

The horizontal arrows again show the relationship of matching logic's concepts, as outlined in \Cref{sec:background}, highlighting how applications and the sorts used in satisfaction checking are governed by the signature. The application interpretations used for semantic evaluation no longer fit on the horizontal arrows, so they are connected via a branch from the relevant arrows. Vertical arrows show the functions outlined in this paper. Since the concrete values are displayed, the function names were moved to the arrows.

\begingroup
\scriptsize
\centering
\begin{tikzpicture}[
    remember picture,
    circled/.style={draw, ellipse},
    every node/.append style={inner sep=.5pt},
    bullet/.style={fill, circle, inner sep=1pt},
    maybebox/.style={draw},
    modelmatrix/.style={matrix of math nodes, below=12mm, nodes={anchor=base}, maybebox},
    every matrix/.append style={ampersand replacement=\&}
]
    \matrix[matrix of nodes, nodes in empty cells] (m) {
        \&[1.5mm] \&[3.5mm] Theory \&[1mm] \node[anchor=base] {\clap{\shortstack{Semantic\\evaluation}}}; \&[2mm] Carrier elements \&[-1mm] \node[anchor=base] {{\shortstack{Satisfaction\\checking}}}; \\[10mm]
        \dict{nat} \& |[above] (sym1)| \dict{plus} \& $\exists \ld \mu \ld \bevar0 \lor \syapp {\tikzmarknode[circled]{circ1}{\dict{plus}}} {\sephetlist{\bevar0, \bsvar0}}$ \& |[bullet]| {} \& $\bigcup\limits_{x \in \dict{metanat}} \app{\lfp}{(\la {Y} \{x\} \cup \{x + y \mid y \in Y\})}$ \& $\stackrel?= \dict{metanat}$ \\
        \dict{bool} \& |[below] (sym3)| \dict{isEven} \& $\exists \ld \syapp {\tikzmarknode[circled]{circ2}{\dict{isEven}}} {\sephetlist{\bevar0}}$ \& |[bullet]| {} \& $\bigcup\limits_{x \in \dict{metanat}} \smallcase{\{\dict{metatrue}\}}{x\text{ is even}}{\{\dict{metafalse}\}}$ \& $\stackrel?= \dict{metabool}$ \\[38mm]
        \dict{gnum} \& |[above] (sym4)| \dict{gplus} \& $\exists \ld \mu \ld \bevar0 \lor \syapp {\tikzmarknode[circled]{circ3}{\dict{gplus}}} {\sephetlist{\bevar0, \bsvar0}}$ \& |[bullet]| {} \& $\bigcup\limits_{x \in \dict{metanat}} \app{\lfp}{(\la {Y} \{x\} \cup \{x + y \mid y \in Y\})}$ \& $\stackrel?= \dict{metanat}$ \\
        \dict{gbool} \& |(sym6)| \dict{gisEven} \& $\exists \ld \syapp {\tikzmarknode[circled]{circ4}{\dict{gisEven}}} {\sephetlist{\bevar0}}$ \& |[bullet]| {} \& $\bigcup\limits_{x \in \dict{metanat}} \smallcase{\{\dict{metatrue}\}}{x\text{ is even}}{\{\dict{metafalse}\}}$ \& $\stackrel?= \dict{metabool}$ \\
        \dict{gcbool} \& |[below] (sym8)| \dict{gcisEven} \& $\syapp {\tikzmarknode[circled]{circ5}{\dict{gcisEven}}} {\sephetlist{\fevar x}} \to \syapp {\tikzmarknode[circled]{circ6}{\dict{gcisEven}}} {\sephetlist{\syapp {\tikzmarknode[circled]{circ7}{\dict{gthrice}}} {\sephetlist{\fevar x}}}}$ \& |[bullet]| {} \& $\left(\tinynegspace\dict{metanat} \setminus \smallcase{\{1\}}{n\text{ is even}}{\{0\}}\right) \cup \left(\tinynegspace\smallcase{\{1\}}{3 \cdot n\text{ is even}}{\{0\}}\right)$ \& $\stackrel?= \dict{metanat}$ \\[45mm]
        pos \& |[above] (sym9)| \dict{add} \& $\exists \ld \mu \ld \bevar0 \lor \syapp {\tikzmarknode[circled]{circ8}{\dict{add}}} {\sephetlist{\bevar0, \bsvar0}}$ \& |[bullet]| {} \& $\bigcup\limits_{x \in \dict{metapos}} \app{\lfp}{(\la {Y} \{x\} \cup \{x + y \mid y \in Y\})}$ \& $\stackrel?= \dict{metapos}$ \\
        bool \& |[below] (sym11)| \dict{isEven} \& $\syapp {\tikzmarknode[circled]{circ9}{\dict{isEven}}} {\sephetlist{\fevar x}} \to \syapp {\tikzmarknode[circled]{circ10}{\dict{isEven}}} {\sephetlist{\syapp {\tikzmarknode[circled]{circ11}{\dict{thrice}}} {\sephetlist{\fevar x}}}}$ \& |[bullet]| {} \& $\left(\tinynegspace\dict{metanat} \setminus \smallcase{\{1\}}{n\text{ is even}}{\{0\}}\right) \cup \left(\tinynegspace\smallcase{\{1\}}{3 \cdot n\text{ is even}}{\{0\}}\right)$ \& $\stackrel?= \dict{metanat}$ \\
    };

    \node at ($(m-1-1)!.5!(m-1-2)$) {Signature};
    \node (sym2) at ($(sym1)!.5!(sym3)$) {\dict{twice}};
    \node (sym5) at ($(sym4)!.5!(sym6)$) {\dict{gtwice}};
    \node (sym7) at ($(sym6)!.5!(sym8)$) {\dict{gthrice}};
    \node (sym10) at ($(sym9)!.5!(sym11)$) {\dict{thrice}};
    \node[above=7mm] at (m-1-3) {};

    \matrix[modelmatrix] at (m-3-4) (model1) {
        \syapp{\dict{plus}}{\sephetlist{n, m}} \mapsto \{n + m\} \\
        \syapp{\dict{twice}}{\sephetlist{n}} \mapsto \{2 \cdot n\} \\
        \syapp{\dict{isEven}}{\sephetlist{n}} \mapsto \smallcase{\{\dict{metatrue}\}}{n\text{ is even}}{\{\dict{metafalse}\}} \\
    };
    \matrix[modelmatrix] at (m-6-4) (model2) {
        \syapp{\dict{gplus}}{\sephetlist{n, m}} \mapsto \{n + m\} \\
        \syapp{\dict{gtwice}}{\sephetlist{n}} \mapsto \{2 \cdot n\} \\
        \syapp{\dict{gisEven}}{\sephetlist{n}} \mapsto \smallcase{\{\dict{metatrue}\}}{n\text{ is even}}{\{\dict{metafalse}\}} \\
        \syapp{\dict{gthrice}}{\sephetlist{n}} \mapsto \{3 \cdot n\} \\
        \syapp{\dict{gcisEven}}{\sephetlist{n}} \mapsto \smallcase{\{1\}}{n\text{ is even}}{\{0\}} \\
    };
    \matrix[modelmatrix] at (m-8-4) (model3) {
        \syapp{\dict{add}}{\sephetlist{n, m}} \mapsto \{n + m\} \\
        \syapp{\dict{thrice}}{\sephetlist{n}} \mapsto \{3 \cdot n\} \\
        \syapp{\dict{isEven}}{\sephetlist{n}} \mapsto \smallcase{\{1\}}{n\text{ is even}}{\{0\}} \\
    };
        
    \draw[-Stealth] (circ1.north) to[in=0, out=90, looseness=.25] (sym1.east);
    \draw[-Stealth] (circ2.south) to[in=0, out=270, looseness=.25] (sym3.east);
    \draw[-Stealth] (circ3.north) to[in=0, out=90, looseness=.25] (sym4.east);
    \draw[-Stealth] (circ4.north) to[in=25, out=90, looseness=.25] (sym6.north east);
    \draw[-Stealth] (circ5.south) to[in=0, out=270, looseness=.25] (sym8.east);
    \draw[-Stealth] (circ6.south) to[in=0, out=270, looseness=.25] (sym8.east);
    \draw[-Stealth] (circ7.north) to[in=0, out=90, looseness=.25] (sym7.east);
    \draw[-Stealth] (circ8.north) to[in=0, out=90, looseness=.25] (sym9.east);
    \draw[-Stealth] (circ9.south) to[in=0, out=270, looseness=.25] (sym11.east);
    \draw[-Stealth] (circ10.south) to[in=0, out=270, looseness=.25] (sym11.east);
    \draw[-Stealth] (circ11.north) to[in=0, out=90, looseness=.25] (sym10.east);

    \draw[-Stealth] (m-2-6.north) to[in=45, out=135, looseness=.15, edge label'=$\dict{nat} \mapsto \dict{metanat}$] (m-2-1.north east);
    \draw[-Stealth] (m-3-6.south) to[in=315, out=225, looseness=.15, edge node={node[above, xshift=-2mm] {$\dict{bool} \mapsto \dict{metabool}$}}] (m-3-1.south east);
    \draw[-Stealth] (m-4-6.north) to[in=45, out=135, looseness=.25, edge label'=$\dict{gnum} \mapsto \dict{metanat}$] (m-4-1.north east);
    \draw[-Stealth] (m-5-6.north) to[in=45, out=135, looseness=.35, edge label'=$\dict{gbool} \mapsto \dict{metabool}$] (m-5-1.north east);
    \draw[-Stealth] (m-6-6.south) to[in=315, out=225, looseness=.20, edge node={node[above, xshift=-4mm] {$\dict{gcbool} \mapsto \dict{metanat}$}}] (m-6-1.south east);
    \draw[-Stealth] (m-7-6.north) to[in=45, out=135, looseness=.15, edge label'=$\dict{pos} \mapsto \dict{metapos}$] (m-7-1.north east);
    \draw[-Stealth] (m-8-6.south) to[in=315, out=225, looseness=.20, edge node={node[above, xshift=-2mm] {$\dict{bool} \mapsto \dict{metanat}$}}] (m-8-1.south east);

    \draw[-Stealth] (m-2-3.base east) to (m-2-5.base west);
    \draw[-Stealth] (m-3-3.base east) to (m-3-5.base west);
    \draw[-Stealth] (m-4-3.base east) to (m-4-5.base west);
    \draw[-Stealth] (m-5-3.base east) to (m-5-5.base west);
    \draw[-Stealth] (m-6-3.base east) to (m-6-5.base west);
    \draw[-Stealth] (m-7-3.base east) to (m-7-5.base west);
    \draw[-Stealth] (m-8-3.base east) to (m-8-5.base west);

    \draw[-Stealth] (m-2-4) to (model1);
    \draw[-Stealth] (m-4-4) to (model2);
    \draw[-Stealth] (m-7-4) to (model3);

    \node[maybebox, fit={(m-2-1) (m-3-1) (sym1) (sym2) (sym3)}] (sig1) {};
    \node[maybebox, fit={(m-4-1) (m-5-1) (m-6-1) (sym4) (sym5) (sym6) (sym7) (sym8)}] (sig2) {};
    \node[maybebox, fit={(m-7-1) (m-8-1) (sym9) (sym10) (sym11)}] (sig3) {};
    \node[maybebox, fit={(m-2-3) (m-3-3)}] (axes1) {};
    \node[maybebox, fit={(m-4-3) (m-5-3) (m-6-3)}] (axes2) {};
    \node[maybebox, fit={(m-7-3) (m-8-3)}] (axes3) {};
    \node[maybebox, fit={(m-2-5) (m-3-5)}] (sem1) {};
    \node[maybebox, fit={(m-4-5) (m-5-5) (m-6-5)}] (sem2) {};
    \node[maybebox, fit={(m-7-5) (m-8-5)}] (sem3) {};
    \node[maybebox, fit={(m-2-6) (m-3-6)}] (sat1) {};
    \node[maybebox, fit={(m-4-6) (m-5-6) (m-6-6)}] (sat2) {};
    \node[maybebox, fit={(m-7-6) (m-8-6)}] (sat3) {};

    \draw let \p1 = ([xshift=-2mm] sig2.west) in (sig3.west) to (\p1 |- sig3.west) to[edge node={node[above, sloped] {\dict{glue}}}] (\p1 |- sig1.west) to (sig1.west) (\p1) edge[-Stealth] (sig2.west);
    \draw let \p1 = ([xshift=-2mm] axes2.west) in (axes3.west) to (\p1 |- axes3.west) to[edge node={node[above, sloped, xshift=-15mm] {\dict{glueT}}}] (\p1 |- axes1.west) to (axes1.west) (\p1) edge[-Stealth] (axes2.west);
    \draw let \p1 = ([xshift=1mm] sem2.east) in (model1.east) to (\p1 |- model1.east) to[edge node={node[above, sloped, xshift=-8mm] {\dict{glueM}}}] (\p1 |- model3.east) to (model3.east) (\p1 |- model2.east) edge[-Stealth] (model2.east);
    \draw let \p1 = ([xshift=2mm] sat2.east) in (sat1.east) to (\p1 |- sat1.east) to[edge node={node[above, sloped] {via \Cref*{prop:satisfaction}}}] (\p1 |- sat3.east) to (sat3.east) (\p1) edge[-Stealth] (sat2.east);
\end{tikzpicture}\par
\endgroup
\end{addtoappendix}

\subsection*{Thoughts on model extension}

In this section we discussed the gluing operation in-depth. Its sibling operation, extension, can be defined by altering our approach to gluing, or by simply reducing it to gluing, constructing a signature, a theory, and a model from the necessary elements (e.g., axioms to be added, new and existing sorts and symbols required by them), and utilizing the same methods presented here. We provide a more detailed discussion of these in \Cref{app:modex}.

\begin{addtoappendix}
\section{Thoughts on model extension}\label{app:modex}

While gluing is a balanced operation that combines two structures such that the result contains exactly their union, extension is an unbalanced operation where only a part of such a structure is embedded into another. The addition of new elements in the extension should not invalidate existing satisfaction proofs. During extension, there are three main components that we may wish to extend: sorts, symbols, and theories. All of these present unique challenges that can be resolved in different ways.

\paragraph{Specification-based solution.}
We can implement extension in two ways. In the first, sorts can be extended by simply relaxing the no junk restriction on the partial isomorphisms used for signature combining (\Cref{def:sigcombinecond}), allowing for new sorts to be included. This makes the ``union'' be an arbitrary upper bound of the parameters, rather than the least. However, this breaks all definitions reliant on this axiom, by introducing a third case where a parameter sort is not in the input signature(s). These cases have to be resolved externally, via user-supplied functions determining the variables and the carrier set of the new sorts.

To extend the set of symbols, we must eliminate its no junk condition. As with removing this restriction over sorts, the functions operating or symbols also need to be supplemented externally. The parameter and return sorts, as well as the interpretation of the new symbols must be separately provided by the user. Note that these functions have access to the extended sorts.

With these two, the type class containing the conditions of signature combining (\Cref{def:sigcombinecond}) is revised as follows.

\medskip
\typedef[class]{\app{\dict{extendable}}{(\Sigma_1, \Sigma_2 : \texttt{Signature})}}{
	{\dict{SharedSort}}{\dict{Type}},
	{\dict{subone}}{\app{\dict{partial}}{(\app{\dict{Sorts}}{\Sigma_1}), \dict{SharedSort}}},
	{\dict{subtwo}}{\app{\dict{partial}}{(\app{\dict{Sorts}}{\Sigma_2}), \dict{SharedSort}}},
    {\dict{supplone}}{\fa {s} \app{\dict{partialg}}{\dict{subone}, s} = \dict{None} \to \app{\dict{partialg}}{\dict{subtwo}, s} = \dict{None} \to \dict{Type}},
    {\dict{suppltwo}}{\fa {s} \app{\dict{partialg}}{\dict{subone}, s} = \dict{None} \to \app{\dict{partialg}}{\dict{subtwo}, s} = \dict{None} \to \dict{Type}},
	{\dict{SharedSymbols}}{\dict{Type}},
	{\dict{subfive}}{\app{\dict{partial}}{(\app{\dict{symbols}}{\Sigma_1}), \dict{SharedSymbols}}},
	{\dict{subsix}}{\app{\dict{partial}}{(\app{\dict{symbols}}{\Sigma_2}), \dict{SharedSymbols}}},
    {\dict{supplthree}}{\fa {\sigma} \app{\dict{partialg}}{\dict{subfive}, \sigma} = \dict{None} \to \app{\dict{partialg}}{\dict{subsix}, \sigma} = \dict{None} \to \listof{\dict{SharedSort}}},
    {\dict{supplfour}}{\fa {\sigma} \app{\dict{partialg}}{\dict{subfive}, \sigma} = \dict{None} \to \app{\dict{partialg}}{\dict{subsix}, \sigma} = \dict{None} \to \dict{UnitedSort}}
}
\medskip

The conditional typeclass for models is adjusted accordingly. With the no junk axioms no longer providing totality for the functions defined based upon them, and the newly supplied values requiring a new case, all such functions must be redefined as well. As an example, we present the \dict{params} field provided by such a newly defined \dict{extend} function:

\begin{equation*}
    \app{\dict{params}}{(\app{\dict{extend}}{G}), \sigma} \coloneqq
    \begin{cases}
        \map{(\app{\dict{partialf}}{(\app{\dict{subone}}{G})})}{(\app{\dict{params}}{\Sigma_1, \sigma'})} &\text{if }
        \begin{aligned} &\te {\sigma'} \app{\dict{partialg}}{(\app{\dict{subfive}}{G}), \sigma} = \breakrow\app{\dict{Some}}{\sigma'}\end{aligned} \\
        \map{(\app{\dict{partialf}}{(\app{\dict{subtwo}}{G})})}{(\app{\dict{params}}{\Sigma_2, \sigma'})} &\text{if }
        \begin{aligned} &\te {\sigma'} \app{\dict{partialg}}{(\app{\dict{subsix}}{G}), \sigma} = \breakrow\app{\dict{Some}}{\sigma'}\end{aligned} \\
        \app{\dict{supplthree}}{\sigma, H_1, H_2} &\text{if }
        \begin{aligned} &(H_1 : \app{\dict{partialg}}{(\app{\dict{subfive}}{G}), \sigma} = \breakrow\dict{None}) \land {} \\ &(H_2 : \app{\dict{partialg}}{(\app{\dict{subsix}}{G}), \sigma} = \breakrow\dict{None})\end{aligned} \\
    \end{cases}
\end{equation*}

All other cases of all such functions are defined in the same manner.

Finally, extending theories with an arbitrary axiom can be done by simply expressing it as a pattern in the combined signature, joining it to the existing theori(es) and providing separate satisfaction proofs for them as before. However, often these axioms are not arbitrary, but variations of existing axioms for the new sorts. These families of axioms can typically be proven the same way. By capturing these proof schemes, these specialized cases of theory extension can be done automatically.

This operation can be scrutinized for certain properties as well. For example, when performing two separate extensions, their order should not matter. Furthermore, extending a glued structure should be the same as gluing two identically extended structures (known as distributivity). In this approach, these properties must be proven separately.

\paragraph{Gluing-based solution.}

The alternative approach is reducing extension to gluing. The new sorts and symbols can be used to define a signature: the sorts can be joined to the existing ones as a disjoint union, as described at the beginning of \Cref{operators}, making the related properties trivial to prove; next, since all sorts are available, the parameter and return sorts of the new symbols can be provided in the usual way. Similarly, the functions defining the carrier sets and interpretations for the new elements can just be grouped into a model and glued as before.

There is no such shortcut for theories; they must be joined the same way in both approaches. The properties, however, are trivially provable as instances of the gluing properties: reordering extensions is just commutativity and associativity, while distributivity follows from those two and idempotence.
\end{addtoappendix}

\section{Related work}

The combination of matching logic models has been previously explored in the context of LTL model checking~\cite{LTLMC}. To combine the models, the authors also utilize extension and gluing; however, semantic consistency is guaranteed by restricting axioms to the ``open fragment'' of the logic, where patterns cannot quantify over elements outside their initial signature. This ensures that the validity of patterns is preserved when new elements are added to a model. However, the applicative, unsorted version of the logic~\cite{chen2021explained} utilized in this work~\cite{LTLMC} handles certain structural concepts in an extrinsic, axiomatic manner. For example, definedness and injection are introduced via dedicated symbols that are expected to be part of every signature and are treated specially.

This is in contrast to the polyadic, sorted matching logic that our work is based on and that handles sorted quantifiers natively. This eliminates the need for introducing auxiliary symbols and the concept of an open fragment, although working with the dependently typed, well-sorted syntax (utilizing heterogeneous-lists) becomes substantially more difficult.

The main drawback in the existing literature is that model gluing requires the models to have disjoint carrier sets, allowing them to handle overlapping elements by simply constructing the new model via a disjoint union. To overcome the disjointness limitation, we draw on frameworks from universal logic and category theory that explicitly handle the combination of logics sharing structural elements~\cite{decpres,meetcomb,algeblogic}.

To translate and lift patterns between local and combined signatures,~\cite{decpres} heavily utilizes two bijections that map formulas by replacing unknown structural elements with dedicated variables. This allows patterns to be lifted seamlessly into a greater combined language while preserving consistency. Defining a similar union for the models requires semantic combination. The model-theoretic counterpart to uniting theories over a shared overlap is the \emph{Amalgamation Property}~\cite{amalgam}, which ensures that two models can be embedded into a single structure without altering the behavior of their overlap.

\emph{Institution Theory}~\cite{hybridins} may be utilized to bridge semantic amalgamation with the syntactic union of signatures. Research in this area establishes that a \emph{pushout square} of signature morphisms (the syntactic union of signatures) corresponds directly to an \emph{amalgamation square} for models. This fundamental relationship provides the category-theoretical guarantee that when models are combined, they naturally satisfy the theories translated into the combined signature.

Rather than relying on abstract syntactic translations or enforcing disjoint carrier sets, we introduce a constructive, operational framework for these unions. We compute the union of types by explicitly handling overlaps via partial isomorphisms. In place of variable substitution, we utilize a partial order to lift patterns into a greater signature. Applying these universal algebraic tools to sorted, polyadic matching logic provides a novel method for combining non-disjoint signatures and corresponding models.

\section{Conclusion and future work}

In this paper, we have established a rigorous theoretical foundation for model composition within a polyadic, many-sorted variant of matching logic. We defined the systematic composition of signatures, variable valuations, and constituent models, and argued that essential satisfaction proofs are preserved within the resulting composite architecture. Furthermore, by formulating our definitions using dependent types, we have provided a blueprint for imminent machine-checked formalization and future automation.

The mechanisms described in this paper are lacking a machine-checked formalization. The immediate next step of this research is creating such an implementation using the existing formalization of the logic~\cite{kore-ml}, and machine-checking our results. Furthermore, we believe that given the right infrastructure, such as a global collection of known isomorphic types, it would be possible to create automation for instantiating the type classes described in this paper to create systematic model definitions suitable for given composite theories, potentially originating from the $\mathbb K$ framework.

\subsection*{Acknowledgments}

We would like to thank Máté Tejfel for his comments and advice in the paper writing process. This research was supported by a generous grant received from Pi Squared Inc. Its contents are solely the responsibility of the authors and do not necessarily represent the official views of the entities providing funding for said research.

\clearpage


\renderappendix
\end{document}